%% file: main.tex
\documentclass[pdflatex,sn-nature]{sn-jnl}

\usepackage{graphicx}%
\usepackage{multirow}%
\usepackage{amsmath,amssymb,amsfonts}%
\usepackage{amsthm}%
\usepackage{mathrsfs}%
\usepackage[title]{appendix}%
\usepackage{xcolor}%
\usepackage{textcomp}%
\usepackage{manyfoot}%
\usepackage{booktabs}%
\usepackage{algorithm}%
\usepackage{algorithmicx}%
\usepackage{algpseudocode}%
\usepackage{listings}%
\usepackage{physics}
\usepackage{bibunits}
\defaultbibliographystyle{unsrt}
\defaultbibliography{reference}

\theoremstyle{thmstyleone}%
\newtheorem{theorem}{Theorem}
\theoremstyle{thmstyletwo}%

\theoremstyle{thmstylethree}%

\begin{document}
\newgeometry{left=1.25in, right=1.25in, top=0.9in, bottom=1in}
\title{Powers of Magnetic Graph Matrix: Fourier Spectrum, Walk Compression, and Applications}


\author[1]{Yinan Huang}

\author[2]{David F. Gleich}

\author*[1]{Pan Li}\email{panli@gatech.edu}

\affil[1]{School of Electrical and Computer Engineering, Georgia Institute of Technology, Atlanta, Georgia, USA}

\affil[2]{Department of Computer Science, Purdue University, West Lafayette, Indiana, USA}


\maketitle

\begin{center}
\begin{minipage}{0.8\textwidth}
\begin{center}
\subsubsection*{Abstract}
\end{center}
\vspace{0.5em}

Magnetic graphs, originally developed to model quantum systems under magnetic fields, have recently emerged as a powerful framework for analyzing complex directed networks. Existing research has primarily used the spectral properties of the magnetic graph matrix to study global and stationary network features. However, their capacity to model local, non-equilibrium behaviors, often described by matrix powers, remains largely unexplored. We present a novel combinatorial interpretation of the magnetic graph matrix powers through directed walk profiles---counts of graph walks indexed by the number of edge reversals. Crucially, we establish that walk profiles correspond to a Fourier transform of magnetic matrix powers. The connection allows exact reconstruction of walk profiles from magnetic matrix powers at multiple discrete potentials, and more importantly, an even smaller number of potentials often suffices for accurate approximate reconstruction in real networks. This shows the empirical compressibility of the information captured by the magnetic matrix. This fresh perspective suggests new applications; for example, we illustrate how powers of the magnetic matrix can identify frustrated directed cycles (e.g., feedforward loops) and can be effectively employed for link prediction by encoding local structural details in directed graphs.

\end{minipage}
\vspace{2em}
\end{center}


\begin{bibunit}

\input{1_Intro}
\input{2_Walk_profile}

\input{3_Compressed}

\input{4_Motif}
\input{5_Link}
\input{6_Conclusion}

\putbib
\end{bibunit}

\begin{bibunit}
\clearpage
\input{7_Method}
\renewcommand{\refname}{References for Methods}
\putbib
\end{bibunit}

\begin{bibunit}    
\clearpage
\input{8_Supplementary}
\putbib
\end{bibunit}
\end{document}

%% file: 1_Intro.tex
Magnetic graphs are complex-weighted graphs in which each edge is assigned a complex phase. Originally developed to model the effects of external magnetic fields on electrons in lattice systems~\cite{harper1955single, shubin1994discrete, pankrashkin2006spectra, colin2013magnetic, berkolaiko2013nodal, razzoli2020continuous}, magnetic graphs have recently proven to be powerful tools for analyzing complex systems with asymmetric interactions~\cite{guo2017hermitian, mohar2020new, li2022hermitian, bottcher2024complex}. These graphs encode edge directionality through complex phases in their matrix representations, such as adjacency or Laplacian matrix. For instance, the spectrum of magnetic graph Laplacians has been used to reveal community structures in directed cycles~\cite{fanuel2017magnetic} and to enhance visualization of directed networks~\cite{fanuel2018magnetic, f2020characterization}. More recently, magnetic graphs have inspired novel signal processing and machine learning models tailored to directed graphs~\cite{zhang2021magnet, geisler2023transformers, huang2024good}. Additionally, magnetic graph Laplacians can be interpreted as connection graph Laplacians associated with the SO(2) group~\cite{cloninger2024random}, enabling further applications in angular synchronization~\cite{bandeira2013cheeger, singer2011angular} and high-dimensional data analysis via diffusion maps~\cite{singer2012vector}.

Previous studies on the magnetic graph matrix have primarily focused on how their eigenvalues and eigenvectors encode the topological properties of directed graphs~\cite{fanuel2017magnetic, fanuel2018magnetic, f2020characterization, lange2015frustration}. Such spectral information is widely recognized for capturing global and stationary characteristics of graphs~\cite{chung1997spectral}. In contrast, modeling local or non-equilibrium behaviors often requires the analysis of random walks, quantum walks, or diffusion processes~\cite{lovasz1993random, chung1997spectral, venegas2012quantum}, which are often characterized by the powers of the graph matrix. For undirected graphs, these processes have been extensively studied and proven effective for detecting local communities and node clusters~\cite{pons2005computing, rosvall2008maps, mucha2010community, lambiotte2015random}, quantifying node centrality (e.g., PageRank~\cite{page1999pagerank, haveliwala2002topic}, Katz index~\cite{katz1953new}) and measuring node proximity~\cite{brand2005random, yildirim2008random}. However, their counterparts on magnetic graphs remain relatively underexplored. This stems not only from the complex-valued nature of the magnetic matrix, whose powers resist straightforward application of standard Markov process theory~\cite{fanuel2017magnetic}, but also critically, from the lack of a clear combinatorial interpretation that links the matrix power to the underlying directed network topology.

    In this article, we investigate the powers of the magnetic graph matrix to uncover properties of local network behaviors encoded in directed graph topologies. Our primary finding is that the $m$-th power of the magnetic graph adjacency matrix with magnetic potential $q$ naturally corresponds to the Fourier transform of a bidirectional walk statistic, which we term the \emph{Walk Profile}. Specifically, the $m$-step walk profile consists of bidirectional walk counts of length $m$, indexed by the number of edge reversals. These ideas also give probabilistic interpretations when using the normalized adjacency matrix, in which case, these walk counts become the corresponding bidirectional random walk landing probabilities. This result is significant as it reveals a fundamental connection between the behaviors on magnetic graphs captured by powers of the magnetic matrix and the combinatorial structure of the underlying directed graphs, described by walk profiles.  It also offers a novel physical interpretation of the magnetic potential $q$ as the Fourier transform frequency. Importantly, this connection implies that by evaluating magnetic matrix powers at a finite set of discrete potentials, one can reconstruct the entire $m$-step bidirectional walk profiles via inverse Fourier transforms, thereby enabling the expression of magnetic matrix powers at arbitrary values of potential $q$.

In theory, to exactly reconstruct the $m$-step walk profiles, evaluating magnetic matrix powers at $\lfloor m/2 \rfloor+1$ distinct frequencies $q$ is required. However, for most random and real-world graphs, we demonstrate that significantly fewer frequencies often provide an accurate approximate reconstruction, as the energy of the walk profiles predominantly concentrates in the low-frequency region. The accuracy of this compression depends on the spectral radius of the magnetic graph matrix, which is related to the smallest eigenvalue of the magnetic graph Laplacian and its associated frustration index~\cite{lange2015frustration}. This relationship enables us to characterize directed graphs whose walk profiles are either easily compressible or resistant to compression, by analyzing the eigenvalues of their corresponding magnetic graph matrix.

Regarding practical applications, we further explore the feasibility of capturing diverse directed motifs using magnetic matrix powers with a varying number of potentials. In particular, we demonstrate that these representations can detect \emph{minimally frustrated directed cycles}, i.e., directed cycles perturbed by a single reversed edge, which extend beyond purely directed cycles emphasized in previous spectral analyses of magnetic Laplacians~\cite{fanuel2017magnetic}. The minimally frustrated directed cycles include feedforward loops, a prominent three-node motif prevalent in control systems and biological networks~\cite{mangan2003structure}. Finally, we employ magnetic matrix powers as structural features for directed graphs in link prediction tasks, aimed at inferring missing directed edges in incomplete networks. The strong performance of these features highlights their effectiveness in encoding local structural information of directed graphs.

%% file: 2_Walk_profile.tex
\vspace{8pt}
\subsubsection*{Magnetic graph matrix and walk profiles}
\begin{figure}[t!]
    \centering
    \includegraphics[width=1\linewidth]{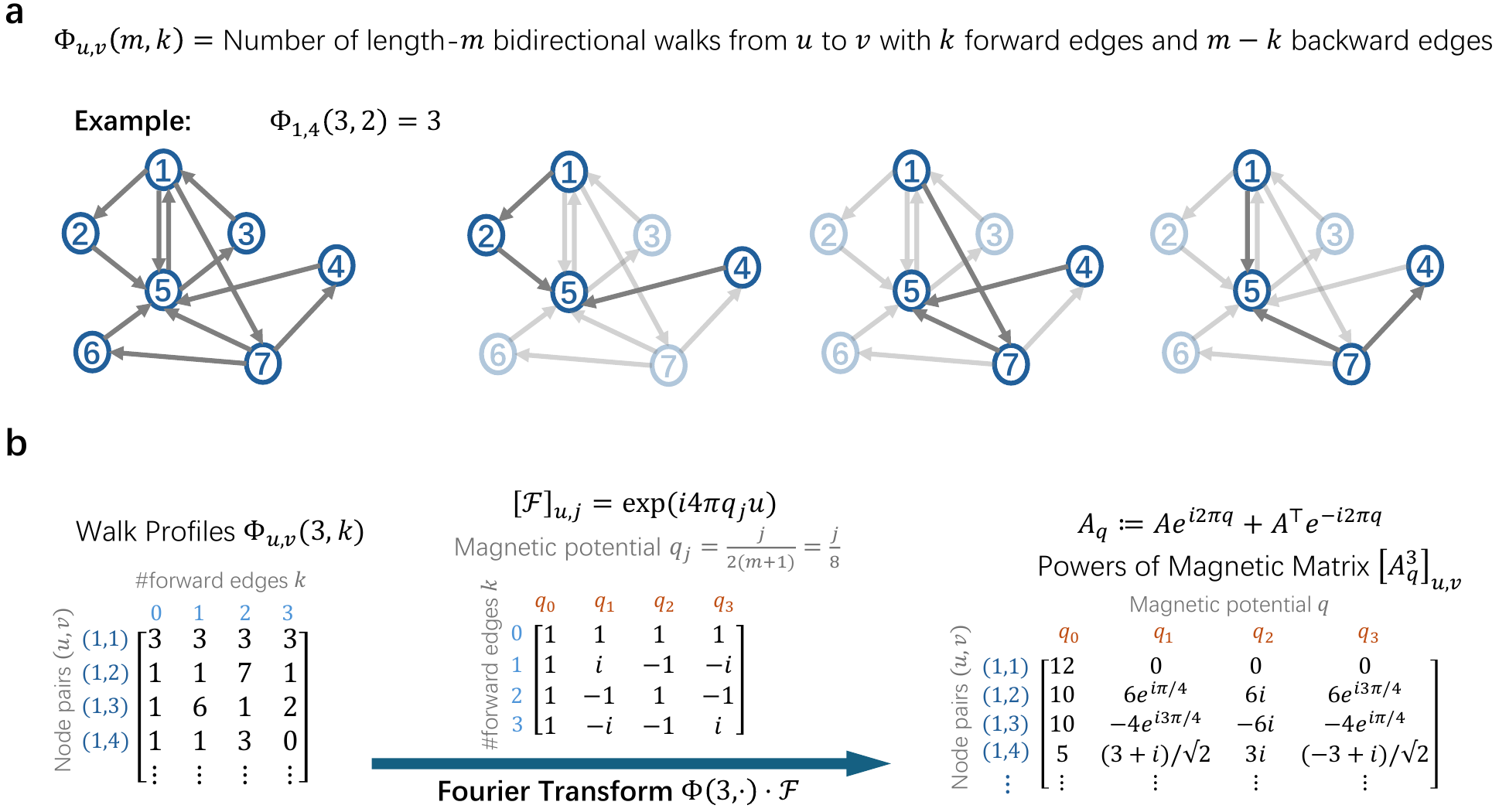}
    \caption{\textbf{Illustrations of walk profiles and its Fourier transforms.} \textbf{a}, An example of computing walk profiles $\Phi_{1,4}(3,2)$, the number of $3$-step bidirectional walks from node $1$ to node $4$ with $2$ forward edges and $1$ backward edges. The graph is given in the left figure and the three repetitions show the three different paths counted in $\Phi_{1,4}(3,2)$.  \textbf{b}, The collection of walk profiles $\Phi_{u,v}(3,k), k=0,...,3$ on the same graph, and how its Fourier transform relates to the powers of the magnetic adjacency matrix.}
    \label{fig:intro-figure}
\end{figure}

We first define \textbf{Walk Profiles} that captures bidirectional walk patterns of a directed graph, and we will then show how powers of the magnetic matrix are inherently connected to walk profiles (see Figure~\ref{fig:intro-figure}). 
Let $\mathcal{G}$ be a directed graph with asymmetric adjacency matrix $A$, where $A_{u,v}=1$ if $u$ connects to $v$, denoted by $u\to v$ and otherwise $A_{u,v}=0$. If $\mathcal{G}$ is a weighted directed graph, then $A_{u,v}$ is real-valued and represents the edge weight. 
A bidirectional walk refers to a sequence of nodes $u_1,u_2,...,u_m$, where each transition step can be either be forward edge $u_{i}\to u_{i+1}$ or a backward edge $u_{i}\leftarrow  u_{i+1}$. Walk profiles, denoted by $\Phi_{u,v}(m,k)$, are defined by the number of $m$-step bidirectional walks from node $u$ to node $v$, with exact $k$ forward edges and $m-k$ backward edges. An example is illustrated in Figure~\ref{fig:intro-figure}(a). These can be computed via a general iterative matrix expression (see eq.~\eqref{eq:wp_recursive} in the supplement), for instance, $\Phi_{u,v}(2,2)=[A^2]_{u,v}, \Phi_{u,v}(2,1)=[AA^{\top}+A^{\top}A]_{u,v}, \Phi_{u,v}(2,0)=[(A^{\top})^2]_{u,v}$. 

The directed graph $\mathcal{G}$ induces a magnetic graph with potential $q$, defined via a complex-valued, Hermitian adjacency matrix $A_q=Ae^{i2\pi q}+A^{\top }e^{-i2\pi q}$ called magnetic adjacency matrix~\cite{mohar2020new}. A forward edge $u\rightarrow v$ induces a positive-phase edge weight $[A_q]_{u,v}=e^{i2\pi q}$ and a backward edge $u\leftarrow v$ induces a negative-phase edge weight $[A_q]_{u,v}=e^{-i2\pi q}$, so that the phases encode the direction of edges in $\mathcal{G}$. 
We focus on the powers of magnetic adjacency matrix $A_q^m$ that naturally characterizes how information propagates on the magnetic graph. 



One of our key findings is that by expanding the powers of matrix power $[A_q^m]_{u,v}$  we can rewrite it as a Fourier transform of walk profiles $[\Phi_{u,v}(m, m),...,\Phi_{u,v}(m,0)]$ at frequency $2q$: 
\begin{equation}
    [A_q^m]_{u,v}=[(Ae^{i2\pi q}+A^{\top}e^{-i2\pi q})^m]_{u,v}=e^{-i2\pi qm}\sum_{k=0}^me^{i4\pi qk}\Phi_{u,v}(m,k).
    \label{eq:ft}
\end{equation}
It is convenient to think walk profiles in a vector form: $\Phi(m,\cdot)=(\Phi(m,m), \Phi(m,m-1),...,\Phi(m,0))$. We can rewrite equation \eqref{eq:ft} as ${A}_q^m=e^{-i2\pi qm}\cdot\mathcal{F}_{2 q}(\Phi(m,\cdot))$, i.e., $A^m_q$ is the discrete Fourier transform of a signal $\Phi(m,\cdot)$ at ordinary frequency $2q$. In this sense, $A_q^m$ is a projection of walk profiles $\Phi(m,\cdot)$ into the frequency subspace of frequency $2q$. For example if $q=0$, $A_{q=0}^m=(A+A^{\top})^m$ is the simple sum $\sum_{k=0}^m\Phi(m,k)$, which accounts for walks in all possible directions. See Supplemental section~\ref{supp:fourier-result} for the full details. 



\subsubsection*{Magnetic graph matrix with multiple potentials}
The Fourier transform interpretation implies that $A_{q}^m$ only carries partial information about $\Phi(m,\cdot)$ and it suggests that we consider a full set of potentials $\vec{q}=(q_0,...,q_{Q-1})$ so that $A^m_{\vec{q}}=(A^m_{q_0},...,A^m_{q_{Q-1}})$ can express full information in $\Phi(m,\cdot)$. This is important because once $\Phi(m,\cdot)$ is reconstructed from $A_{\vec{q}}^m$, then $A_q^{m}$ for arbitrary $q$ can be calculated by equation \eqref{eq:ft}. In other words, there exists a finite set of potentials $\vec{q}$ that capture all possible information magnetic adjacency matrix can express. Note that equation (1) with $Q$ many potentials $\vec{q}$ forms a $Q\times (m+1)$ linear system:
\begin{equation}
    \begin{pmatrix}
       e^{i2\pi q_0m}\cdot A^m_{q_0}\\ e^{i2\pi q_1m}\cdot A_{q_1}^m\\\cdots\\ e^{i2\pi q_{Q-1}m}\cdot A_{q_{Q-1}}^m 
    \end{pmatrix}=
    F_{\vec{q}, m+1}\begin{pmatrix}
        \Phi(m,m)\\\Phi(m,m-1)\\\cdots\\\Phi(m,0)
    \end{pmatrix},
    \label{eq:linear_system}
\end{equation}
where $[F_{\vec{q},Q}]_{j ,k}=e^{i4\pi q_j k}$ is a $Q\times (m+1)$ Fourier matrix.
Since $\Phi(m,\cdot)$ is real, we can compare real parts and imaginary parts of the linear system and construct a $2Q\times (m+1)$ equivalent linear system of real matrices. A canonical choice to solve a Fourier system is to let $\vec{q}=(q_0,...,q_{\lfloor m/2\rfloor+1})$ with $q_{j}\triangleq \frac{j}{2(m+1)}$, which leads to the inverse Fourier transform:
\begin{equation}
    \Phi(m,k)=\frac{1}{m+1}\left(A_{q_0}^m+2\cdot\sum_{j=1}^{\lfloor m/2\rfloor}\mathrm{Re}(e^{\frac{i\pi j(m-2k)}{m+1}}\cdot A^m_{q_j})\right)
    \label{eq:linear_system_solution}
\end{equation}
Therefore, using potentials with $q_{j}=\frac{j}{2(m+1)}$, $j=0,1,...,\lfloor m/2\rfloor$ allows us to fully capture walk profiles. More generally, choosing any arbitrary $\lfloor m/2\rfloor+1$ distinct potentials in the range $[0, \frac{1}{4}]$ can also solve the linear system, although an analytical solution like equation \eqref{eq:linear_system_solution} may not exist.

Note that all the results we present so far are not exclusively for $A_q$, but also can be adapted to other types of magnetic graph matrix, such as random walk magnetic matrix $D^{-1}A_q$ and normalized magnetic adjacency matrix $D^{-1/2}A_qD^{-1/2}$ ($D=\text{Diag}(d_1,d_2,...)$ is the diagonal degree matrix, counting both in-edges and out-edges). For example, let us denote $\widehat{A}_q\triangleq D^{-1}A_q$ and $\widehat{\Phi}(m,k)$ as the corresponding random-walk version. Equations \eqref{eq:ft}, \eqref{eq:linear_system} and \eqref{eq:linear_system_solution} still hold if we replace all $\Phi(m,k)$ by $\widehat{\Phi}(m,k)$ and all $A_q^m$ by $\widehat{A}_q^m$. Particularly, when we consider the random-walk version of walk profiles, $[\hat{\Phi}(m,k)]_{u,v}$ can be interpreted as the transition probability of $m$-step bidirectional random walks from node $u$ to $v$, whose random walk sequences contain exact $k$ forward edges and $m-k$ backward edges. The probabilistic interpretation implies the normalization property $\sum_{k=0}^m\sum_{v}[\widehat{\Phi}(m,k)]_{u,v}=1$.

%% file: 3_Compressed.tex
\subsubsection*{Spectral sparsity and compressibility of walk profiles}
\begin{figure}[t!]
    \centering
    \includegraphics[width=1\linewidth]{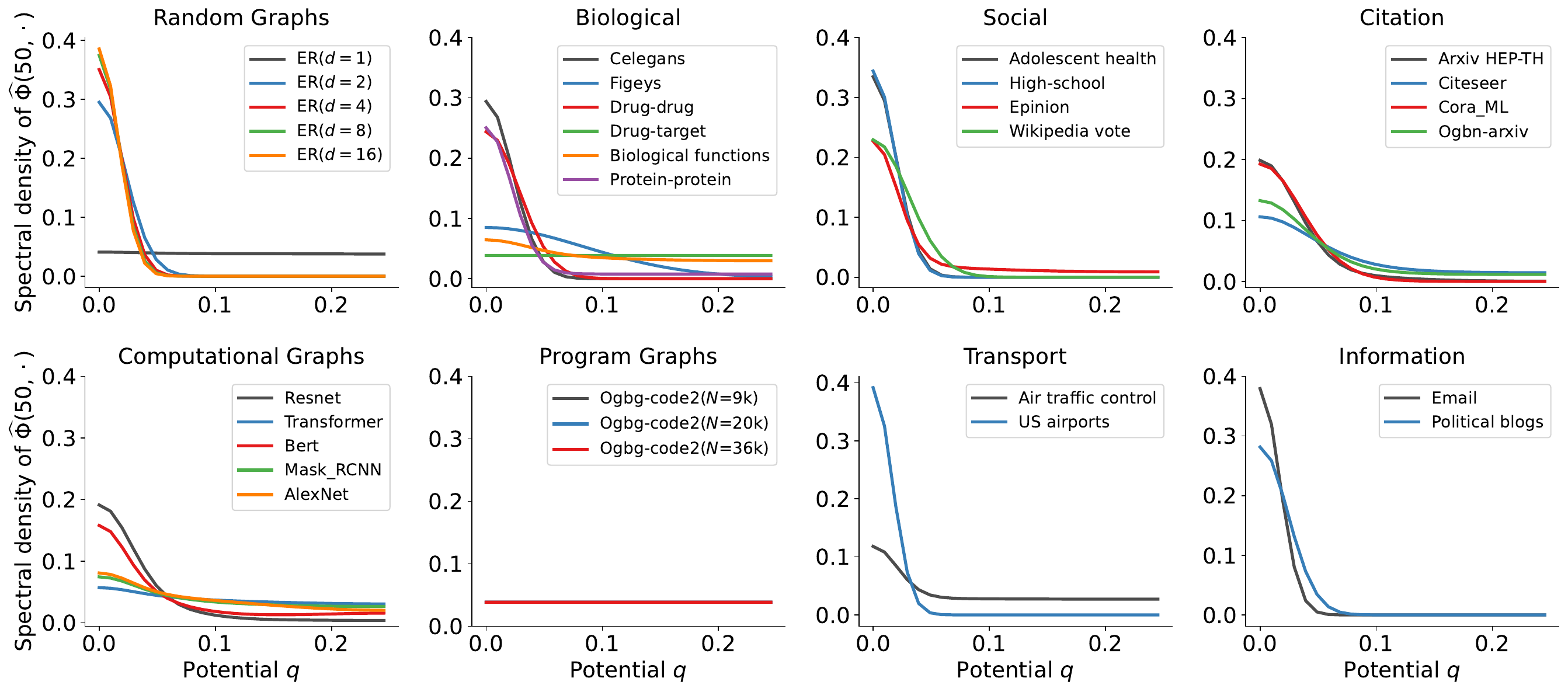}
    \caption{\textbf{Spectral density of walk profiles.} The average relative spectral density of walk profiles $\Phi(50, \cdot)$, i.e., $\langle |A_q^{50}|^2\rangle$ with respect to potential $q$, often shows a concentration in the low-frequency regime. The average $\langle \cdot\rangle$ is taken over all node pairs in the network.}
    \label{fig:spectral_density_q}
\end{figure}
\begin{figure}[t!]
    \centering
    \includegraphics[width=1\linewidth]{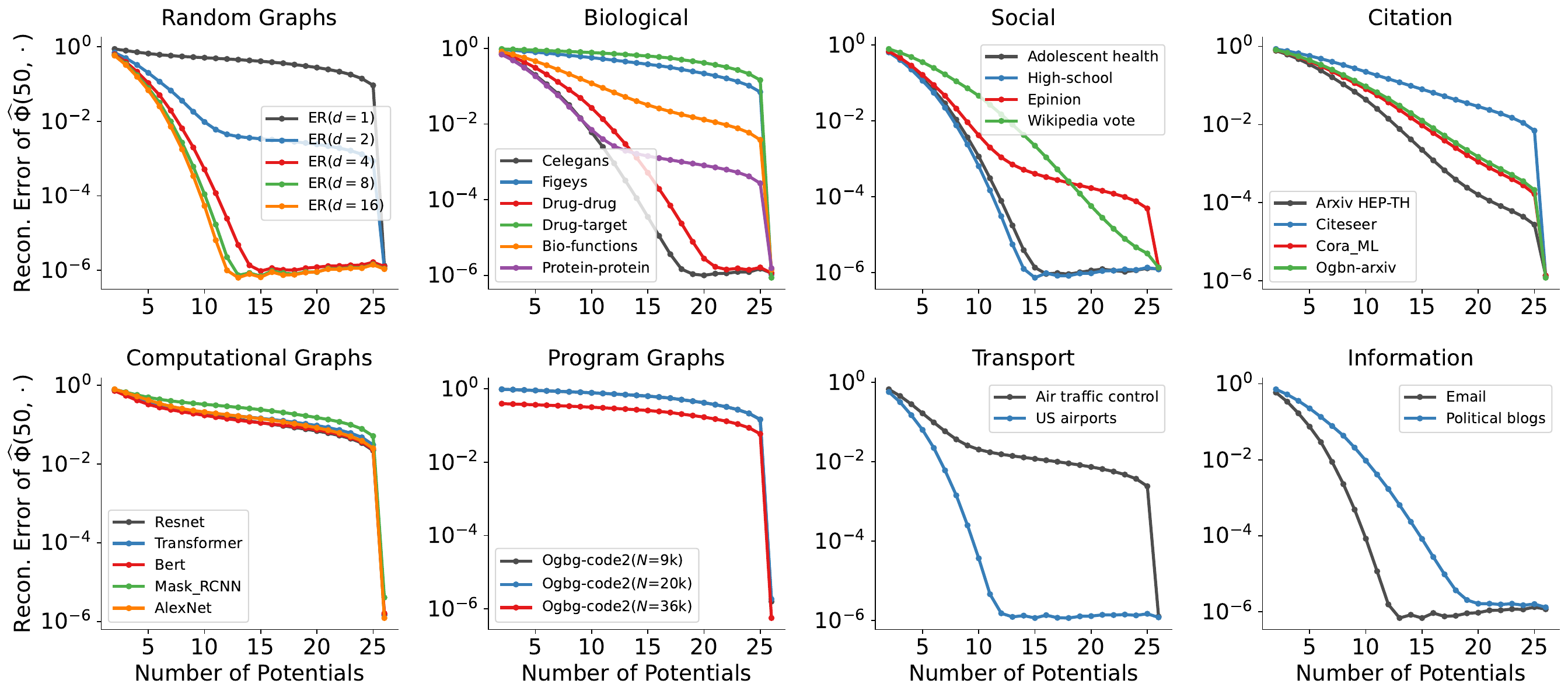}
    \caption{\textbf{Compressibility of walk profiles.} The average relative reconstruction error of walk profiles $\widehat{\Phi}(50, \cdot)$ computed using $[\widehat{A}_{q_0}^{50}, ..., \widehat{A}^{50}_{q_{Q-1}}]$ with varying numbers of the top-$Q$ smallest potentials, on a wide range of directed networks from different domains.}
    \label{fig:recon_error_q}
\end{figure}

Having established the walk profile as the Fourier transform of the magnetic matrix, a natural question arises: to what extent do all magnetic potentials (frequencies) contribute meaningfully to walk profiles? Many biological and physical systems exhibit sloppiness, meaning their behavior is largely governed by a few stiff combinations of parameters, while being relatively insensitive to variations in the remaining, sloppy directions~\cite{gutenkunst2007universally, transtrum2010nonlinear}. Similarly, in classical signal processing, many natural signals can be efficiently approximated using only a small number of dominant frequencies~\cite{candes2006robust, elad2010sparse, bruckstein2009sparse, chen2016signal}. These motivates us to ask whether walk profiles admit a similar kind of compressibility.


Here we explore the spectral density of random-walk walk profiles $\widehat{\Phi}_{u,v}(m,\cdot)$, i.e., the magnitude $|[\widehat{A}_q^m]_{u,v}|$ with respect to different potential $q$'s (formal definitions in Methods). We analyze Erd\H{o}s-Renyi random graphs and 26 real-world directed networks from 8 different domains, ranging in size from $N\approx 10^{2}$ to $N\approx 10^{6}$ (Methods and Supplementary Information Section~\ref{supp:datasets}).  We mainly adopt walk length $m=50$, while our discoveries and conclusions are consistent for different walk lengths (Supplementary Information Section~\ref{supp:compression}). 
Figure~\ref{fig:spectral_density_q} shows that on most networks, the energy in walk profiles tends to be largely concentrated on the low-frequency region. 


This strongly implies that walk profiles, the combinatorial structure of bidirectional walks, are compressible: although $\lfloor m/2\rfloor+1$ many potentials are required to precisely express walk profiles $\widehat{\Phi}(m,\cdot)$ in the frequency space, the top-$Q$ smallest potentials already carry the essential information of walk profiles and should be able to approximate walk profiles with neglectable reconstruction errors. To verify this, we choose the first $Q$ smallest potentials $\vec{q}=(\frac{0}{2(m+1)}, \frac{1}{2(m+1)},...,\frac{Q-1}{2(m+1)})$ and reconstruct normalized walk profiles $\widehat{\Phi}(m,\cdot)$ from $\widehat{A}_{\vec{q}}^m$ by solving the undertermined linear system (equation \eqref{eq:linear_system}). Figure \ref{fig:recon_error_q} show how reconstruction error (the relative error between true $\widehat{\Phi}(m,\cdot)$ and the estimated one, averaged over node pairs) varies with different number of potentials $Q$. As expected, as we increases $Q$, the reconstruction error gradually decreases, and goes to zero when $Q$ reaches $\lfloor m/2\rfloor+1$. Notably, for many real-world networks, adopting less than $m/5$ number of potentials is sufficient to yield a low reconstruction error of $5\%$. 
This finding, however, is far from universal and the data show substantial variation in reconstruction. For instance, computational and program graphs require more potentials than social and information networks. We study this behavior next. 

\subsubsection*{Spectral sparsity from decayed spectral radius}
\begin{figure}
    \centering
    \includegraphics[width=1\linewidth]{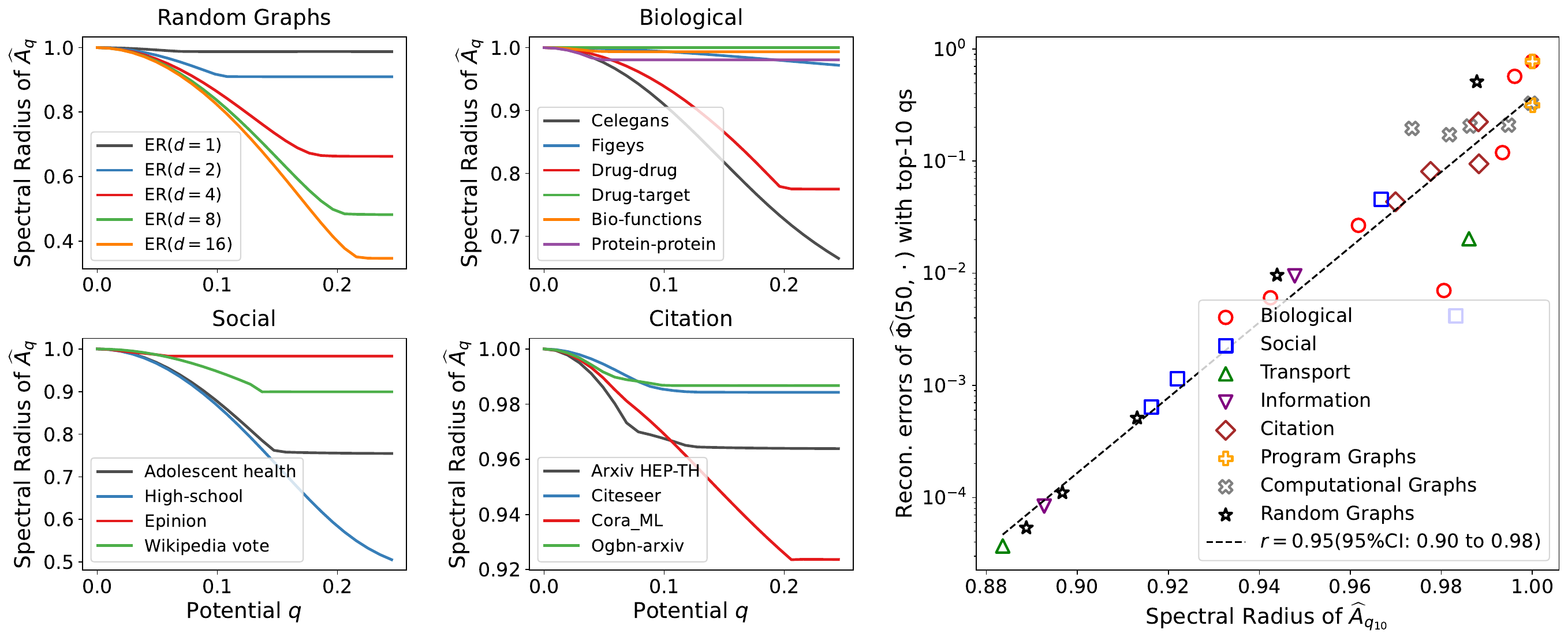}
    \caption{\textbf{Spectral radius of $\widehat{A}_q$ characterizes walk profiles compressibility.} Left four subfigures: the decay phenomenon of spectral radius of $\widehat{A}_q$, w.r.t. the potential $q$ on some representative networks. Right: The decay rate of spectral radius (measured by spectral radius of $\widehat{A}_{q_{10}}$ with $q_{10}=\frac{10}{2(50+1)}=5/51$) is strongly correlated with the compressibility of walk profiles $\widehat{\Phi}(50,\cdot)$, with Pearson correlation coefficient $r=0.95$ and 95\% confidence interval $(0.90, 0.98)$ estimated by Fisher transformation.}
    \label{fig:recon_error_spectral_radius}
\end{figure}

We have found the degree of energy concentration on low-frequency region, and thus the compressibility of walk profiles, varies across different networks. Here, we argue that this spectral sparsity phenomenon stems from and is governed by the global structure of directed networks, specifically, the spectral radius of random-walk magnetic matrix $\widehat{A}_q$. Intuitively, the magnitude $|[\widehat{A}_{q}^m]_{u,v}|^2$ is asymptotically governed by the spectral radius $\rho(\widehat{A}_q)$. Thus, we expect that if $\rho(\widehat{A}_q)<\rho(\widehat{A}_{q^{\prime}})$, then $|[\widehat{A}^m_{q}]_{u,v}|\ll |[\widehat{A}^m_{q^{\prime}}]_{u,v}|$ as $m$ goes large. Note that the spectral radius of $\widehat{A}_q$ is closely related to the smallest eigenvalue of magnetic Laplacian $\widehat{L}_q=I-\widehat{A}_q$, which is known to describe how ``consistent'' the magnetic graph is, measured by the maximal accumulated phases over any cycles in the directed graph~\cite{fanuel2018magnetic}.

The left subfigures of Figure \ref{fig:recon_error_spectral_radius} illustrate how the spectral radius $\rho(\widehat{A}_q)$ varies with $q$ (more results in Supplementary Information Section~\ref{supp:decayed_spectral_radius}). Interestingly, we consistently observe a decayed pattern of $\rho(\widehat{A}_q)$ as $q$ grows, which corresponds to the low-frequency structure of walk profiles. To further validate the connection between the spectral radius and walk profile compressibility, we compare the reconstruction errors of using $Q$ potentials vs the spectral radius of $\widehat{A}_{q_{Q}}$ with $q_{Q}=\frac{Q}{2(m+1)}$ (here we choose $m=50$ and $Q=10$). If the spectral radius of $\widehat{A}_{q_{Q}}$ is small, it should yield a stronger energy concentration effect and thus a better compressibility. Indeed, a clear positive correlation between these two is observed (Pearson correlation coefficient $r=0.95$ with 95\% confidence interval $[0.9,0.98]$), as shown in the right subfigure in Figure \ref{fig:recon_error_spectral_radius}. This also explains why the compressibility is poor for program graphs, since they are directed trees with $\rho(\widehat{A}_q)\equiv 1$, leading to the uniform spectral density of their walk profiles. Supplementary Information Section~\ref{supp:hard_examples} further supports the correlation between compressibility and spectral radius, showing that compression becomes progressively easier for graphs with a monotonically decreasing spectral radius.


%% file: 4_Motif.tex
\subsubsection*{Directed network motif detection}

\begin{figure}
    \centering
    \includegraphics[width=1\linewidth]{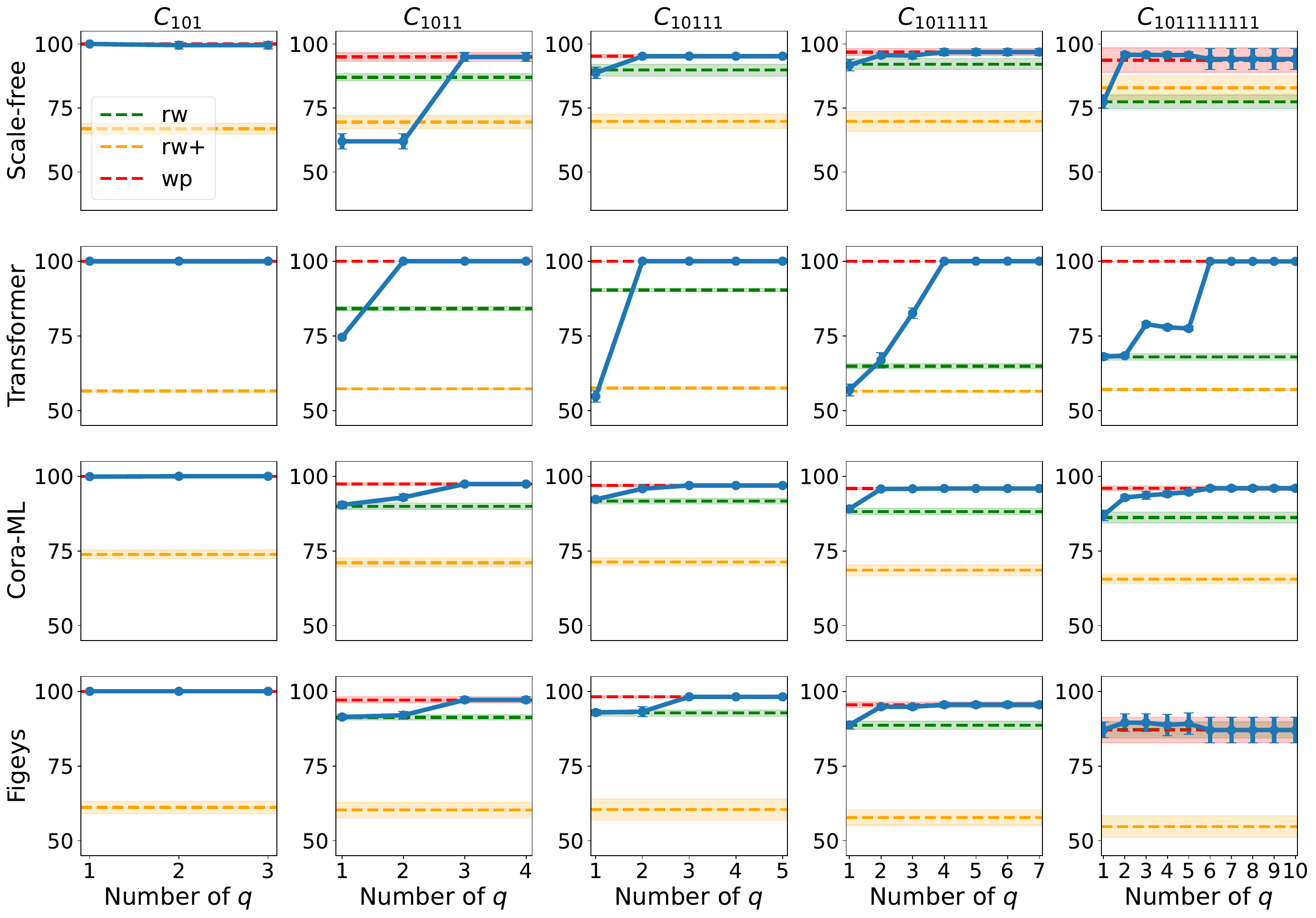}
    \caption{\textbf{Motif detection.} Each row corresponds to different datasets (networks), and each column corresponds to cycles with one backward edge of different length. A binary string $101\cdots 1$ is used to represent a general cycle $C_{101\cdots 1}$. Here $101\cdots1$ refers to the ordering of forward and backward edges along the cycle, e.g., $C_{101}$ means $u_1\xrightarrow{1} u_2\xleftarrow{0} u_3\xrightarrow{1} u_1$. In each subfigure, the blue line represents the the test AUC (Y axis) for detecting the target cycle using powers of Magnetic matrix, with varying number of potentials (X axis). We also compare it to undirected walks (rw), unidirected walks (rw+) and walk profiles (wp). The error band stands for one standard deviation for 10 runs of experiments with different random seeds.}
    \label{fig:motif_detection}
\end{figure}
 We now explore the ability to detect motifs in directed networks using the powers of magnetic adjacency matrix and walk profiles. Network motifs are substructures (subgraphs) that serve as the fundamental building blocks of complex networks, which can be used to uncover a network's functional abilities, e.g., for biological networks~\cite{milo2002network,wong2012biological, masoudi2012building}. Particularly, we are interested in direction-aware cycles, i.e., a sequence of unrepeated nodes $(u_1,u_2,...,u_m,u_1)$ that can be traversed with either forward edges or backward edges. Given a cycle length, there are different non-isomorphic classes of cycles, categorized by the order of the edge orientations. For instance, 4-cycles have four classes: fully oriented $u_1\rightarrow u_2
\rightarrow u_3\rightarrow u_4\rightarrow u_1$,  partially oriented with one backward edge $u_1\leftarrow u_2\rightarrow u_3\rightarrow u_4\rightarrow u_1$, partially oriented with two successive backward edges $u_1\leftarrow u_2\leftarrow u_3\rightarrow u_4\rightarrow u_1$, partially oriented with two nonsuccessive backward edges $u_1\leftarrow u_2\rightarrow u_3\leftarrow u_4\rightarrow u_1$.

We are interested in if the powers of magnetic adjacency matrix $([A^m_{q_0}]_{u,u},...,[A^m_{q_{Q-1}}]_{u,u})$ can identify if a node $u$ is a part of a cycle of length $m$. While walk-based methods generally cannot guarantee the detection of cycles, because walks permit repeated nodes whereas cycles do not, they are frequently adopted as a proxy for cycle detection for two main reasons: (a) Walk counts correlate well with the presence or absence of cycles. For instance, a walk count of zero for a specific length $m$ at node $u$ indicates the absence of $m$-cycles involving $u$. Moreover, a walk that does not revisit any intermediate nodes is, in fact, a cycle; 
  (b) Calculating walks is computationally efficient, primarily requiring standard matrix multiplication.

Remarkably, through an analysis of their theoretical capacity (Supplementary Information Section~\ref{supp:motif}), we can prove that even assuming walks have no repeated nodes, it is impossible for walk profiles or magnetic adjacency matrix powers to detect cycles containing more than one backward edge. This fundamental limitation arises because walk profiles only account for the number of backward edges, not their specific ordering. Consequently, walk profiles cannot differentiate between non-isomorphic cycles that possess the same quantity of backward edges but differ in the arrangement of these edges. This inability to distinguish extends to the powers of magnetic adjacency matrix, as they convey the same informational content as walk profiles, a point substantiated by $\text{Eq.}~\eqref{eq:ft}$.

On the other hand, walk profiles (and powers of magnetic adjacency matrices) can detect fully-oriented cycles and, importantly, one-backward-edge cycles, assuming the absence of repeated node counting in the walks. This is because when there is at most one backward edge in the cycles, the position of the single backward edge does not affect the detection of such substructures, due to the rotational symmetry of cycles. The class of one-backward-edge cycles include feed-forward loops prominent in biological systems~\cite{mangan2003structure}, which cannot be detected by unidirected walks. Notably, this capacity goes beyond prior spectral analyses of magnetic Laplacians~\cite{fanuel2017magnetic}, which primarily focused on purely directed cycles (i.e., cycles with all edges aligned to the same direction).


To validate our argument above, we train and test a linear classifier (logistic regression) to predict if a node $u$ is involved in a class of length-$m$ cycle, using $\Phi_{u,u}(m,\cdot)$ or $[A_{q_0}^m]_{u,u},...,[A_{q_{{Q-1}}}^m]_{u,u}$ as input features. We also compare to undirected walks, i.e., $(A+A^{\top})^m$, and unidirected walks, i.e., $A^m$ and $(A^{\top})^m$. The detailed descriptions of the experimental setup are in Methods. Figure~\ref{fig:motif_detection} shows that walk profiles can effectively capture cycles involving a single backward edge. Meanwhile, the performance of the magnetic adjacency matrix, with an increasing number of potentials, progressively approaches the performance of walk profiles, as anticipated from their theoretical connection.  


%% file: 5_Link.tex
\subsubsection*{Directed network link prediction}
\begin{figure}
    \centering
    \includegraphics[width=1\linewidth]{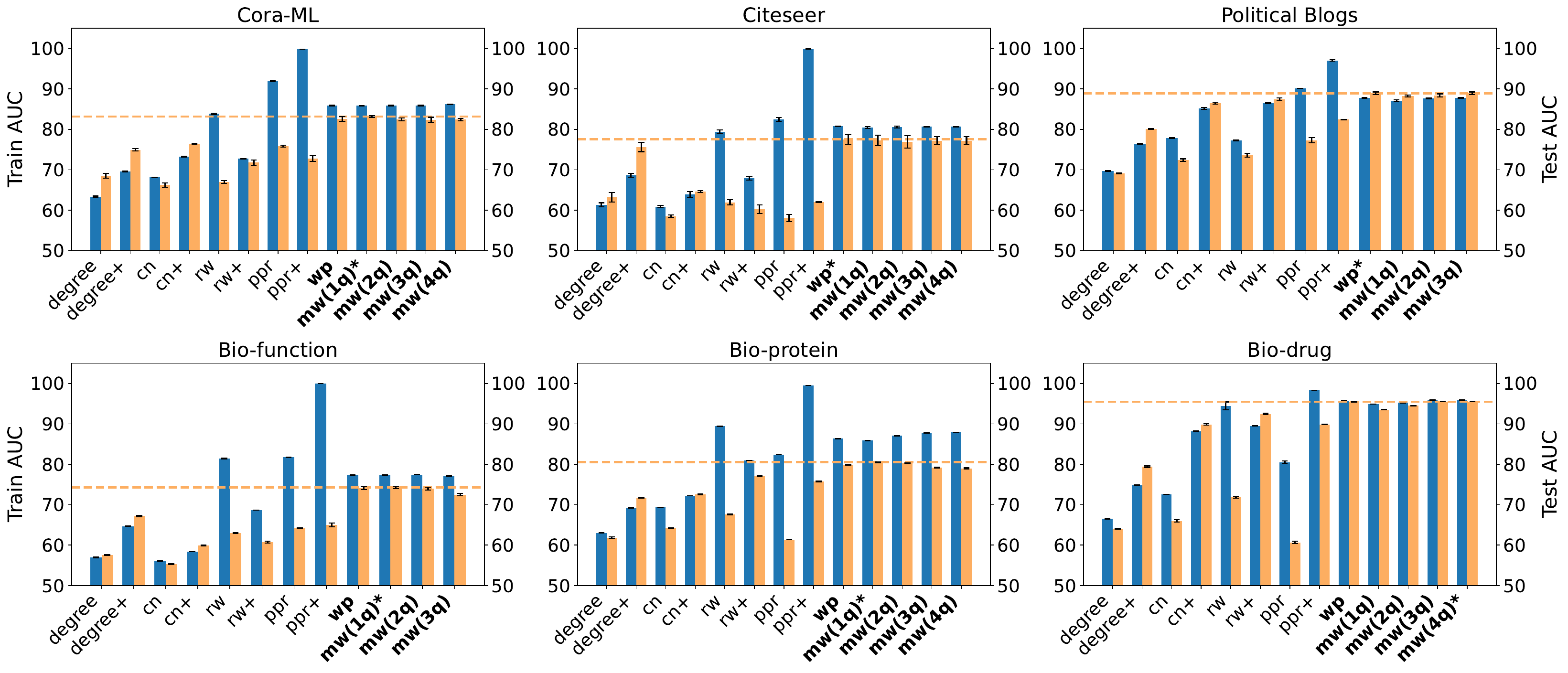}
    \caption{\textbf{Magnetic adjacency matrix and walk profiles are effective for link prediction.} The blue bars represent the training AUCROC scores, and the orange bars represent the test AUCROC scores (higher is better for both). Methods labeled \textbf{mw} refer to the magnetic walk results, and \textbf{wp} denotes the walk profile, both of which are highlighted in bold. Other baseline methods are described in the Methods section. The method with the best test AUCROC is marked with an asterisk, and a horizontal dashed line indicates its test ROC-AUC value. The error bar stands for one standard deviation of 3 runs of experiments with different random seeds.}
    \label{fig:link_prediction}
\end{figure}
We next explore the application of magnetic matrix to link prediction tasks, by leveraging their ability to capture directed structural patterns that can help the inference of missing or future edges. We train and test a linear classifier to predict if an directed edge exists from $u$ to $v$. We use link features $[A_{q_0}^1]_{u,v},...,[A^1_{q_{Q-1}}]_{u,v},...,[A_{q_0}^m]_{u,v},...,[A^m_{q_{Q-1}}]_{u,v}$ representing the walk information within an $m$-hop subgraph. For comparison, we evaluate several baseline methods commonly used for link prediction (refer to Methods). It is worth noticing that the baseline methods are mostly only capable of capturing unidirectional information, i.e., none of them can express bidirectional patterns such as feed-forward loops (3-cycles with one backward edge). Figure~\ref{fig:link_prediction} shows the AUCROC score on the training set and the test set using different link features mentioned. Compared to other baseline methods, both magnetic adjacency matrix and walk profiles consistently achieve a better test prediction performance, highlighting their effectiveness in capturing crucial substructures and predicting missing edges. Remarkably, the powers of magnetic adjacency matrix with only a single potential achieve competitive results to walk profiles, beating all baseline methods. This again suggests the empirical compressibility of walk profiles in the context of inference of missing directed links. 

%% file: 6_Conclusion.tex
\subsubsection*{Discussion}

The issue of directed networks for modeling complex systems with asymmetric interactions has a rich history. Despite its breadth of applications, directed networks have posed theoretical and algorithmic challenges, particularly in spectral and walk-based analyses. This arises due to the tension between the asymmetry of the interaction with the desire for symmetry for node-level properties, along with the inherent non-reversibility of the common associated random walk processes. These asymmetries hinder the direct application of the powerful spectral theory and algorithmic frameworks developed for undirected graphs. To address this, researchers have proposed various symmetric matrix representations for directed graphs, including straightforward symmetrization via 
$A+A^{\top}$, bipartite graph representations~\cite{kleinberg1999authoritative,dhillon2001co,Rohe2016}, bibliometric symmetrization $AA^{\top}+A^{\top}A$~\cite{rudrich2020random}, asymmetric Laplacians~\cite{mohar1997some,boley2011commute}, directed graph Laplacian based on circulations~\cite{chung2005laplacians}, and motif-based symmetrizations~\cite{benson2016higher}. The magnetic Laplacian we study emerged from this study as a flexible approach to ``symmetric'' modeling of directed networks. It preserves the ``symmetry'' of problems through the complex-valued edge weights encoding of edge orientation, yielding a Hermitian matrix representation of the directed networks. Moreover, it admits a natural physical interpretation as modeling the phase acquired by a charged particle (e.g., an electron) in a magnetic field. These properties make the magnetic matrix a compelling framework for understanding the geometry and dynamics of directed networks.

At the same time, there has been a long-standing and fruitful exchange of ideas among physics, combinatorics, algorithms, and algebra in the study of undirected networks. Foundational concepts such as the graph Laplacian, graph cuts, random walks, and the heat (diffusion) equation are deeply interconnected, offering a unified lens for both theoretical analysis and practical algorithms~\cite{chung1997spectral}. While the magnetic matrix provides a clear physical interpretation for directed networks, its connection to combinatorial structure remains less clear. 
This paper resolves that issue by uncovering a new combinatorial foundation for the magnetic matrix via the concept of walk profiles. 
Walk counts are a fundamental concept in network science and these new formulations could enable new relationships with heat kernels and matrix resolvents~\cite{higham2008functions}. They also generalize to non-backtracking random walks~\cite{hashimoto1989zeta, krzakala2013spectral}, which we leave for future studies. 
This connection offers a novel perspective to study complex systems with directional interactions, and thus opens exciting new avenues for theory and applications of complex networks.

%% file: 7_Method.tex
\subsection*{Methods}

\subsubsection*{Network datasets}
The list of all datasets and their statistics are provided in Supplementary Information Section~\ref{supp:datasets}. 

The synthetic datasets in our paper include Erd\H{o}s-R'{e}nyi random graphs~\cite{erdds1959random} (ER graphs) and scale-free random directed graphs~\cite{bollobas2003directed}. For ER graphs, we set number of nodes $N=5000$ and each node pair $(u,v)$ can be an directed edge independently with probability $p=d/N$, with varying $d$. For scale-free random directed graphs, it starts with a directed 3-cycle and sequentially add nodes and edges until it reaches the desired number of nodes $N$. Each step it can: (1) with probability $\alpha$, add a new node connected to an existing node chosen randomly according to the in-degree distribution; (2) with probability $\beta$, add an edge between two existing nodes. One existing node is chosen randomly according the in-degree distribution and the other chosen randomly according to the out-degree distribution; (3) with probability $1-\alpha-\beta$, add a new node connected to an existing node chosen randomly according to the out-degree distribution. We set number of nodes $N=3000$ and $\alpha=0.75, \beta=0.2$.

The real-world datasets covers 26 networks from 8 different domains, ranging in size from $N\approx 10^2$ to $N\approx 10^6$. It contains widely adopted domains such as a) Biological networks (Celegans~\cite{white1986structure}, Figeys~\cite{ewing2007large}, Drug-drug~\cite{wishart2018drugbank}, Drug-target~\cite{zitnik2018modeling}, Bio-functions~\cite{ashburner2000gene}, Protein-protein~\cite{agrawal2018large}); b) Citation networks (Cora\_ML~\cite{mccallum2000automating}, Citeseer~\cite{giles1998citeseer}, ogbn-arxiv~\cite{hu2020open}, Arxiv HEP-TH~\cite{leskovec2005graphs}); c) Social networks (Adolescent health~\cite{massa2009bowling}, High-school~\cite{coleman1964introduction}, Epinion~\cite{richardson2003trust}, Wikipedia vote~\cite{leskovec2010signed}); d) Information networks (Political blogs~\cite{adamic2005political}, Email~\cite{yin2017local}) and e) Transportation networks (Air traffic control~\cite{faa2017}, US airports~\cite{opsahl2011anchorage}). We also explore two emerging domains of interest in directed network analysis: program graphs (ogbg-code2~\cite{hu2020open}) and TPU computational graphs (AlexNet, ResNet, Transformers, Mask-RCNN, BERT)~\cite{phothilimthana2023tpugraphs}.

\subsubsection*{Spectral density of walk profiles}
The spectral density of walk profiles $\Phi_{u,v}(m,\cdot)$, by equation~\eqref{eq:ft}, is $|[A_q^m]_{u,v}|^2$. As we choose discrete frequencies $q_{j}=\frac{j}{2(m+1)}$, we further define the relative spectral density
\begin{equation}
    S_{u,v}(q_j)\triangleq \frac{|[A_{q_j}^m]_{u,v}|^2}{\sum_{\ell =1}^{\lfloor m/2\rfloor +1}|[A_{q_\ell}^m]_{u,v}|^2}.
\end{equation}
For each network dataset, we report the average relative spectral density over a subset of node pairs.  W first randomly sample $N_0< N$ source nodes (denoted by $\mathcal{V}_0\subseteq \mathcal{V}$) from the node set, and then estimate over $N_0\cdot N$ node pairs where the first node belongs to this subset $\mathcal{V}_0$,
\begin{equation}
    \langle S(q_j)\rangle \triangleq \frac{1}{N_0\cdot N}\sum_{u\in\mathcal{V}_0,v\in\mathcal{V}}S_{u,v}(q_j).
\end{equation}
This avoids the lengthy calculation over $N^2$ node pairs $(u,v)$ on large graphs. We found that $N_0=2000$ is sufficient for a good estimate in all networks we adopt, since further increasing $N_0$ does not alter the value of $\langle S(q_j)\rangle$ within our chosen precision (four decimal places). 

\subsubsection*{Walk profile reconstruction algorithm}

The reconstruction of walk profiles $\Phi_{u,v}(m,\cdot)$ from the powers of magnetic matrix is via solving the linear system~Eq.\eqref{eq:linear_system}. In our empirical study, as the walk profiles are always real, we can further expand the linear system in real domain by comparing the real and imaginary parts on both sides:
\begin{equation}
    \begin{pmatrix}
       \mathrm{Re}(e^{i2\pi q_0m}\cdot A^m_{q_0})\\\cdots\\ \mathrm{Re}(e^{i2\pi q_{Q-1}m}\cdot A_{q_{Q-1}}^m) \\\mathrm{Im}(e^{i2\pi q_0m}\cdot A^m_{q_0})\\\cdots\\\mathrm{Im}(e^{i2\pi q_{Q-1}m}\cdot A_{q_Q}^m)
    \end{pmatrix}=
    \begin{pmatrix}
    \mathrm{Re}(F_{\vec{q}, m+1})\\
    \mathrm{Im}(F_{\vec{q}, m+1})
    \end{pmatrix}
    \begin{pmatrix}
        \Phi(m,m)\\\Phi(m,m-1)\\\cdots\\\Phi(m,0)
    \end{pmatrix},
    \label{eq:linear_system_2}
\end{equation}
 To solve this equation, we leverage the Moore–Penrose inverse $(\cdot)^{+}$ of the expanded Fourier matrix (which is of size $2Q$ by $m+1$) to solve $\Phi(m, \cdot)$:
\begin{equation}
    \Phi^{\mathrm{recon}}(m,\cdot)= \begin{pmatrix}
    \mathrm{Re}(F_{\vec{q}, m+1})\\
    \mathrm{Im}(F_{\vec{q}, m+1})
    \end{pmatrix}^{+}\begin{pmatrix}
       \mathrm{Re}(e^{i2\pi q_0m}\cdot A^m_{q_0})\\\cdots\\ \mathrm{Re}(e^{i2\pi q_{Q-1}m}\cdot A_{q_{Q-1}}^m) \\\mathrm{Im}(e^{i2\pi q_0m}\cdot A^m_{q_0})\\\cdots\\\mathrm{Im}(e^{i2\pi q_{Q-1}m}\cdot A_{q_{Q-1}}^m)
    \end{pmatrix}.
    \label{eq:reconstruct_wp}
    \end{equation}
We measure the reconstruction errors by the relative L2 distance defined by
\begin{equation}
    \mathrm{error}=\frac{1}{N\cdot N_0}\sum_{u\in\mathcal{V}_0,v\in\mathcal{V}}\frac{\norm{\widehat{\Phi}_{u,v}(m,\cdot)-\widehat{\Phi}^{\mathrm{recon}}_{u,v}(m,\cdot)}}{\max\{\norm{\widehat{\Phi}_{u,v}(m,\cdot)}, \norm{\widehat{\Phi}^{\mathrm{recon}}_{u,v}(m,\cdot)}\}},
\end{equation}
Again, here we estimate the average error on the same subset of nodes as those in the estimation of the average spectral density.

\subsubsection*{Motif detection method}
The motif detection is a node-level binary classification task to determine if a node is involved in a target motif. We consider the minimally frustrated cycles as the target motif, i.e., directed cycles with one backward edges. Note that we are only interested in whether a node is in such a motif, without considering the possibly different roles of nodes in a motif. For example, consider a feed-forward loop that consists of $u_1\to u_2\to u_3$ and $u_1\to u_3$, then nodes $u_1,u_2,u_3$ are all considered in such a motif. We use the open-sourced library for directed motif searching~\cite{matelsky2021dotmotif} to acquire the ground-truth labels.

Given a directed network, we randomly partition the node set into a training set, a validation set, and a test set. We train a logistic regression on the training set, evaluate the ROCAUC on the validation set, and report the ROCAUC on the test set. We use \texttt{sklearn.linear\_model.LogisticRegression} as the optimizer.

To detect a target cycle of length $m$ at node $u$, the input features to the logistic regression will be one of the four: (1) \texttt{wp}: walk profiles of length $m$ $(\Phi_{u,u}(m,0),...,\Phi_{u,u}(m,m))$; (2) \texttt{mw}: $m$-th powers of magnetic matrix $([A_{q_0}]_{u,u}^m,...,[A_{q_{Q-1}}]_{u,u}^m)$ with $q_{j}=\frac{j}{2(m+1)}$; (3) \texttt{rw}: undirected walks of length $m$ $(A+A^{\top})^m_{u,u}$; (4) \texttt{rw+}: unidirected walks of length $m$ $(A^m_{u,u},(A^{\top})^m_{u,u})$. For powers of magnetic matrix, we first apply the reconstruction mapping (Eq.\eqref{eq:reconstruct_wp}) to obtain $\Phi^{\mathrm{recon}}_{u,u}(m,\cdot)$ for for improving numerical stability. These features are calculated on the whole graph, and each dimension of the input features is normalized by subtracting the mean and dividing by the standard deviation on the training set. We implement both the adjacency-matrix version (i.e., each $A$ here is the adjancency matrix) and the random-walk version (i.e., each $A$ here is replaced by the random walk matrix $D^{-1}A$), and for each target motif detection we report the one with better performance.

\subsubsection*{Link prediction method}
The link prediction task is to predict unobserved edges in an incomplete network. We follow the standard setting of link prediction, randomly partitioning the edge set $\mathcal{E}$ of a directed network into training set $\mathcal{E}_{\mathrm{tr}}$, validation set $\mathcal{E}_{\mathrm{val}}$ and test set $\mathcal{E}_{\mathrm{te}}$. The observed true edges are called positive edges. To perform logistic regression for link prediction, we also need negative samples that are non-edges. To this end, we introduce two types of negative edges: for each $(u,v)\in \mathcal{E}_{\mathrm{tr}}$ (and $\mathcal{E}_{\mathrm{val}}, \mathcal{E}_{\mathrm{te}}$ correspondingly), we add (1) reverse negatives: if $(v,u)\notin \mathcal{E}$, then add $(v,u)$ as a negative edge; (2) structured negatives: for each $(u,v)\in \mathcal{E}$, find $(u, w)\notin \mathcal{E}$ and $(w^{\prime}, v)\notin\mathcal{E}$ with some node $w, w^{\prime}$, and add both of them as negative edges. The final negative edges $\bar{\mathcal{E}}_{\mathrm{tr}}$ for training is formed by randomly choosing negative edges from the selected negative edges described above such that it has the same number of positive edges and the ratio between reverse negatives and two types of structured negatives are 1:1:1. After constructing positive edges $\mathcal{E}_{\mathrm{tr}},\mathcal{E}_{\mathrm{val}},\mathcal{E}_{\mathrm{te}}$ and negative edges $\bar{\mathcal{E}}_{\mathrm{tr}},\bar{\mathcal{E}}_{\mathrm{val}},\bar{\mathcal{E}}_{\mathrm{te}}$, we perform logistic regression trained on dataset $\mathcal{E}_{\mathrm{tr}}$ and $\bar{\mathcal{E}}_{\mathrm{tr}}$, where $(u,v)\in \mathcal{E}_{\mathrm{tr}}$ has label $1$ and label $0$ otherwise. Again we use \texttt{sklearn.linear\_model.LogisticRegression} as the optimizer.

 We consider various baselines as the input features to the logistic regression. To predict if $(u,v)$ is an edge, we use one of the followings: (1) \texttt{wp}: walk profiles up to length $m$ $(\Phi_{u,v}(1,0), \Phi_{u,v}(1,1), \Phi_{u,v}(2,0),...,\Phi_{u,u}(m,m))$; (2) \texttt{mw}: $1,2,...,m$-th powers of magnetic matrix $([A_{q_1}]_{u,v},[A_{q_2}]_{u,v},...,[A_{q_Q}]^m_{u,v})$, with $q_{\ell}=\frac{\ell}{2(m+1)}$; (3) \texttt{rw}: undirected walks up to length $m$  $(A+A^{\top})^m_{u,u}$; (4) \texttt{rw+}: unidirected walks up to length $m$ $(A^m_{u,u},(A^{\top})^m_{u,u})$; (5) \texttt{degree}: node total in-degree + out-degree $(d_{u}, d_{v})$; (6) \texttt{degree+}: node in-degree and out-degree $(d^{\mathrm{in}}_u,d^{\mathrm{out}}_u, d^{\mathrm{in}}_v, d^{\mathrm{out}}_v)$; (7) \texttt{cn}: undirected common neighbor: $N(u)\cap N(v)$ ($N(u)$ means the in-neighbor and out-neighbor of $u$); (8) \texttt{cn+}: directed common neighbor: $(N^{\mathrm{in}}(u)\cap N^{\mathrm{in}}(v), N^{\mathrm{out}}(u)\cap N^{\mathrm{in}}(v), N^{\mathrm{in}}(u)\cap N^{\mathrm{out}}(v), N^{\mathrm{out}}(u)\cap N^{\mathrm{out}}(v))$, where $N^{\mathrm{in}}(u), N^{\mathrm{out}}(u)$ stands for the in-neighbor $\{w: (w, w)\in\mathcal{E}\}$ and out-neighbor $\{w: (u,w)\in\mathcal{E}\}$ respectively; (9) \texttt{ppr}: undirected-version Personalized PageRank: $[(1-\gamma)(1-\gamma D^{-1}(A+A^{\top}))^{-1}]_{u,v}$; (10) \texttt{ppr+}: directed-version Personalized PageRank: $([(1-\gamma)(1-\gamma D_{\mathrm{out}}^{-1}A)^{-1}]_{u,v}, [(1-\gamma)(1-\gamma D_{\mathrm{in}}^{-1}A^{\top})^{-1}]_{u,v})$. We use $\gamma=0.15$ for PageRank throughout our experiments.

Note that the input features to the classifier are calculated differently for different data splits: (1) During training and validation, the input features to predict the link $(u,v)$ are calculated on the graph induced by $\mathcal{E}_{\mathrm{tr}}\backslash\{(u,v)\}$; (2) During testing, the input features are instead calculated on the graph induced by $\mathcal{E}_{\mathrm{tr}}\cup \mathcal{E}_{\mathrm{val}}\backslash\{(u,v)\}$.

\subsubsection*{Code availability}
Our code for reproducing the results is available at https://github.com/Graph-COM/Walk\_Profiles.

\subsubsection*{Acknowledgments}
This work is primarily supported by NSF awards PHY-2117997, IIS-2239565, and IIS-2428777, as well as J.P. Morgan Faculty Award and Meta Research Award.
David F. Gleich would like to acknowledge partial support from NSF IIS-2007481, DOE DE-SC0023162, and the IARPA AGILE program.

%% file: 8_Supplementary.tex
\renewcommand\thesubsection{S.\Roman{subsection}}
\section*{Supplementary Information}
\subsection{Network datasets}
\label{supp:datasets}

\begin{table}[ht!]
\resizebox{0.99\textwidth}{!}{
\begin{tabular}{lllll}
\hline
Category                              & Name                & Description                                                                                                                                                                 & $N$ & $E$ \\ \hline
\multirow{6}{*}{Biological}           & Celegans~\cite{white1986structure}            &           A metabolic network of Caenorhabditis elegans (C. elegans).                                 & 453             & 4596            \\
                                      & Figeys~\cite{ewing2007large}              &           A protein-protein interaction network in Homo sapiens.                                 & 2239            & 6452            \\
                                      & Drug-drug~\cite{wishart2018drugbank}           &            A network of FDA-approved drugs.                & 1514            & 48514           \\
                                      & Drug-target~\cite{zitnik2018modeling}        &           A drug-target interaction network linking drugs to their protein targets.            & 3932            & 18690           \\
                                      & Bio-functions~\cite{ashburner2000gene}       &           A hierarchical network of biological functions defined by Gene Ontology. & 46027           & 106510          \\
                                      & Protein-protein~\cite{agrawal2018large}     &           A protein-protein interaction network of physical interactions in cells.               & 21557           & 342353          \\ \hline
\multirow{4}{*}{Citation}             & Cora\_ML~\cite{mccallum2000automating}            &           A citation network of papers.                                                            & 2995            & 8416            \\
                                      & Citeseer~\cite{giles1998citeseer}            &           A citation data set from an automatic citation indexing system.                                                      & 3312            & 4715            \\
                                      & Ogbn-arxiv~\cite{hu2020open}          &           A citation network of Computer Science arXiv papers indexed by MAG.                                     & 169343         & 1166243         \\
                                      & Arxiv-HEP-TH~\cite{leskovec2005graphs}        &           A citation network of high-energy physics theory papers from arXiv.                            & 27770           & 352807          \\ \hline
\multirow{4}{*}{Social}               & Adolescent health~\cite{massa2009bowling}   &           A friendship network with top five male and female friends.                                & 2539            & 12969           \\
                                      & High-school~\cite{coleman1964introduction}         &           A friendship network between boys in a high-school in Illinois.                                                                                                              & 70              & 366             \\
                                      & Epinions~\cite{richardson2003trust}            &           A directed trust network from Epinions.com.                                               & 75879           & 508837          \\
                                      & Wiki-vote~\cite{leskovec2010signed}           &           A voting network from Wikipedia administrator.       & 7115            & 103689          \\ \hline
\multirow{3}{*}{Program}              & Ogbg-code2 (9k)~\cite{hu2020open}               & \multirow{3}{*}{Abstract syntax trees obtained from Python method definitions.}                                                                 & 9148            & 9147            \\
                                      & Ogbg-code2 (20k)~\cite{hu2020open}    &                                                                                                                                                                                      &   20141              &  20140                \\
                                      & Ogbg-code2 (36k)~\cite{hu2020open}    &                                                                                                                                                                                        &    36123             &     36122            \\ \hline
\multirow{2}{*}{Information}          & Political blogs~\cite{adamic2005political}     &            A directed network of hyperlinks between political weblogs.                                                    & 1224            & 19025           \\
                                      & Email~\cite{yin2017local}               &            An email network among employees at an European research institution.           & 2029            & 39264           \\ \hline
\multirow{2}{*}{Transport}            & Air traffic control~\cite{faa2017} &           A flight network from the U.S. FAA’s National Flight Data Center.               & 1226            & 2615            \\
                                      & USA airports~\cite{opsahl2011anchorage}        &           A directed network of flight routes between U.S. airports.                                     & 1574            & 28236           \\ \hline
\multirow{5}{*}{\parbox{1.8cm}{Computational\\ graphs}} & AlexNet~\cite{phothilimthana2023tpugraphs}             &           \multirow{5}{*}{Tensor Computational Graphs of neural architectures.}                                                                                             & 372             & 597             \\
                                      & ResNet~\cite{phothilimthana2023tpugraphs}              &                                                                                                                                                                                       & 5605            & 9018            \\
                                      & Mask-RCNN~\cite{phothilimthana2023tpugraphs}           &                                                                                                                                                                                       & 14680           & 23604           \\
                                      & Transformer~\cite{phothilimthana2023tpugraphs}         &                                                                                                                                                                                       & 13342           & 21709           \\
                                      & Bert~\cite{phothilimthana2023tpugraphs}                &                                                                                                                                                                                       & 21338           & 37243           \\ \hline
\end{tabular}}
\caption{Description and statistics of directed network datasets. $N$ and $E$ represents the number of nodes and edges.\label{tab:dataset}}
\end{table}

This section describes the network datasets we adopted in our numerical experiments.

\textbf{Synthetic directed networks.} We adopted two types of random directed graph models: Erd\H{o}s-Renyi random graphs~\cite{erdds1959random} (ER graphs) and scale-free random directed graphs~\cite{bollobas2003directed}. For ER graphs, we set number of nodes $N=5000$ and each ordered node pair $(u,v)$ can be an directed edge independently with probability $p=d/N$. We set $d\in \{1, 2, 4, 8, 16\}$.
For scale-free random directed graphs, it starts with a directed 3-cycle and sequentially add nodes and edges until it reaches the desired number of nodes $N$. Each step it can: (1) with probability $\alpha$, add a new node connected to an existing node chosen randomly according to the in-degree distribution; (2) with probability $\beta$, add an edge between two existing nodes. One existing node is chosen randomly according the in-degree distribution and the other chosen randomly according to the out-degree distribution; (3) with probability $1-\alpha-\beta$, add a new node connected to an existing node chosen randomly according to the out-degree distribution. We set number of nodes $N=3000$ and $\alpha=0.75, \beta=0.2$.

\textbf{Real-world directed networks.} Table~\ref{tab:dataset} lists the real-world directed network datasets we adopted in our numerical experiments. The references in the table refer to sources of the original data, and we retrieved the data from the following sources: (1) Celegans, Figyes, Adolescent health, High-school, Political blogs, Email, Air traffic control and USA airports are from~\cite{li2020link}; (2) Cora\_ML and Citeseer are from~\cite{he2024pytorch}; (3) Drug-drug, Drug-target, Bio-functions and Protein-protein are from~\cite{biosnapnets}; (4) Arxiv-HEP-TH, Epinion, Wiki-vote are from~\cite{snapnets}; (5) Ogbn-arxiv, Ogbg-code2 are from~\cite{hu2020open}; (6) AlexNet, ResNet, Mask-RCNN, Transformer and Bert are from~\cite{phothilimthana2023tpugraphs}.

\subsection{Magnetic graph matrix, walk profiles, and their Fourier transform}
\label{supp:fourier-result}

The walk profile $\Phi_{u,v}(m,k)$ is defined by the number of length-$m$ bidirectional walks from node $u$ to node $v$ with $k$ forward edges and $m-k$ backward edges. That is, 
\begin{equation}
    \Phi_{u,v}(m,k)  = \underbrace{[A^{k}(A^{\top})^{m-k}]_{u,v}+[A^{k-1}A^{\top}A(A^{\top})^{m-k-1}]_{u,v}+\cdots +[(A^{\top})^{m-k}A^k]_{u,v}.}_{\text{all matrix products of $m$ matrices, with $k$ of them being $A$ and others being $A^{\top}$}}
\end{equation}
By the definition, there is a iterative or recursive formula to calculate walk profiles:
\begin{equation}
    \Phi(m,k)=\Phi(m-1,k-1)A+\Phi(m-1,k)A^{\top},
\label{eq:wp_recursive}
\end{equation}
where in practice we can further leverage the sparsity of $A$ to speed up the calculation.

Now we are ready to see the powers of magnetic matrix $A_q=Ae^{i2\pi q}+A^{\top}e^{-i2\pi q}$ is naturally described by walk profiles:
\begin{align}
    A_q^{m}&= (Ae^{i2\pi q}+A^{\top}e^{-i2\pi q})^m \\
    &=\underbrace{e^{i2\pi qm}A^m+e^{i2\pi q(m-1)}e^{-i2\pi q}A^{m-1}A^{\top}+e^{i2\pi q(m-1)}e^{-i2\pi q}A^{m-2}A^{\top}A+\cdots e^{-i2\pi qm}(A^{\top})^m}_{\text{binomial expansion}}\\
    &=e^{i2\pi qm}A^m + e^{i2\pi q(m-1)}e^{-i2\pi q}\underbrace{(A^{m-1}A^{\top}+A^{m-2}A^{\top}A+\cdots)}_{\text{group matrix products with the same phase}}+\cdots e^{-i2\pi qm}(A^{\top})^m\\
    &=\sum_{k=0}^m (e^{i2\pi q})^k (e^{-i2\pi q})^{m-k}\cdot \underbrace{\left\{A^{k}(A^{\top})^{m-k}+A^{k-1}A^{\top}A(A^{\top})^{m-k-1}+\cdots +(A^{\top})^{m-k}A^k\right\}}_{\text{all matrix products of $m$ matrices, with $k$ of them being $A$ and others being $A^{\top}$}}\\
    &= \sum_{k=0}^m (e^{i2\pi q})^k (e^{-i2\pi q})^{m-k} \Phi(m,k)\\&=\sum_{k=0}^m e^{i2\pi q(2k-m)}\Phi(m,k)\\
    &=e^{-i2\pi qm}\sum_{k=0}^me^{i4\pi qk}\Phi(m,k). \label{eq:wp2mw}
\end{align}
Therefore, we can interpret $A_q^m$ as the Fourier transform of $\Phi(m,\cdot)=(\Phi(m,m), ..., \Phi(m, 0))$ at frequency $2q$. Since $\Phi(m, \cdot)$ is a length-$(m+1)$ vector, it generally requires $m+1$ different Fourier components $A_{q_0}^m, ..., A^m_{q_{m}}$ to faithfully encode $\Phi(m,\cdot)$. Indeed, we can write the Fourier transform of $\Phi(m,\cdot )$ with multiple magnetic potentials (frequencies) into a linear system:
       \begin{equation}
        \begin{pmatrix}
           e^{i2\pi q_0m}\cdot A^m_{q_0}\\ e^{i2\pi q_1m}\cdot A_{q_1}^m\\\cdots\\ e^{i2\pi q_{Q-1}m}\cdot A_{q_{Q-1}}^m 
        \end{pmatrix}=
        F_{\vec{q}, m+1}\begin{pmatrix}
            \Phi(m,m)\\\Phi(m,m-1)\\\cdots\\\Phi(m,0)
        \end{pmatrix},
    \end{equation}
where $[F_{\vec{q},Q}]_{j ,k}=e^{i4\pi q_j k}$ is a $Q\times (m+1)$ Fourier matrix. In principle, as long as $Q\ge m+1$ and  $q_{j}\neq q_{j^{\prime}}$ (and they do not differ up to a period $1/2$), the linear system is well-posed and thus we can exactly reconstruct walk profiles from powers of magnetic matrix. When $Q=m+1$, a canonical choice is to choose $q_{j}=\frac{j}{2(m+1)}$ so that it forms an orthogonal Fourier basis and $\frac{1}{\sqrt{m+1}}F_{\vec{q}, m+1}$ is a unitary matrix satisfying $\frac{1}{m+1}F_{\vec{q}, m+1}F^{\dagger}_{\vec{q}, m+1}=\frac{1}{m+1}F_{\vec{q}, m+1}^{\dagger}F_{\vec{q}, m+1}=I_{m+1,m+1}$.

We can further reduce the number of magnetic potentials since the walk profiles $\Phi$ are real-valued as long as the directed network has real-valued edge weights. In such a case, we can compare the real and imaginary parts of both sides of the linear equation and obtain
\begin{equation}
    \begin{pmatrix}
       \mathrm{Re}(e^{i2\pi q_0m}\cdot A^m_{q_0})\\\cdots\\ \mathrm{Re}(e^{i2\pi q_{Q-1}m}\cdot A_{q_{Q-1}}^m) \\\mathrm{Im}(e^{i2\pi q_0m}\cdot A^m_{q_0})\\\cdots\\\mathrm{Im}(e^{i2\pi q_{Q-1}m}\cdot A_{q_Q}^m)
    \end{pmatrix}=
    \begin{pmatrix}
    \mathrm{Re}(F_{\vec{q}, m+1})\\
    \mathrm{Im}(F_{\vec{q}, m+1})
    \end{pmatrix}
    \begin{pmatrix}
        \Phi(m,m)\\\Phi(m,m-1)\\\cdots\\\Phi(m,0)
    \end{pmatrix}.
    \label{eq:expand_linear_system}
\end{equation}
This expanded linear system allows us to reduce the required minimal number of magnetic potentials to $\lceil m/2\rceil+1$. Indeed, if $\Phi$ is real, $A^m_q$ has a symmetry property: 
\begin{align}
    A^m_q = (-1)^{2m} A_q^m =  (-1)^m e^{-i\pi m}A_q^m &=(-1)^m e^{-i\pi m}e^{i2\pi qm}\sum_{k}e^{-i4\pi qk}\Phi(m,k)\\
    &= (-1)^m e^{-i2\pi (\frac{1}{2}-q)m}\sum_{k}e^{i4\pi (\frac{1}{2}-q)k}\Phi(m,k)\\
    &=(-1)^m (A^m_{\frac{1}{2}-q})^*.
\end{align}
So, for example, if we let $q_{j}=\frac{j}{2(m+1)}$, then we can get $A^m_{q_{m+1-j}}$ if we already have $A_{q_j}^m$. This symmetry property allows us to derive the inverse Fourier transform from magnetic matrix to walk profiles with only $q_0,...,q_{\lceil m/2\rceil}$:
\begin{equation}
    \Phi(m,k)=\frac{1}{m+1}\left(A_{q_0}^m+\delta_{m,\text{odd}}(-1)^k\cdot i^mA_{q_{\lceil m/2 \rceil   }}^m+2\cdot\sum_{\ell=1}^{\lceil m/2\rceil-1+\delta_{m,\mathrm{even}}}\mathrm{Re}(e^{\frac{i\pi\ell(m-2k)}{m+1}}\cdot A^m_{q_\ell})\right),
\end{equation}
Where $\delta_{m,\text{odd}}=1$ if $m$ is odd and $\delta_{m,\text{odd}}=0$ otherwise. Similar for $\delta_{m,\text{even}}$. From the above equation we can conclude that $\delta_{m,\text{odd}}(-1)^k\cdot i^mA_{q_{\lceil m/2 \rceil   }}^m$ must be zero if $m$ is odd, otherwise it contributes a complex value to real $\Phi$. Thus, whether $m$ is odd or even, the term $\delta_{m,\text{odd}}(-1)^k\cdot i^mA_{q_{\lceil m/2 \rceil}}^m$ is zero. This yields
\begin{equation}
    \Phi(m,k)=\frac{1}{m+1}\left(A_{q_0}^m+2\cdot\sum_{\ell=1}^{\lfloor m/2\rfloor}\mathrm{Re}(e^{\frac{i\pi\ell(m-2k)}{m+1}}\cdot A^m_{q_\ell})\right),
    \label{eq:mw2wp}
\end{equation}
where we use the fact $\lceil m/2\rceil-1+\delta_{m,\mathrm{even}}=\lfloor m/2\rfloor$. As a result, we only need $\lfloor m/2\rfloor +1=\lceil (m+1)/2\rceil$ many potentials.

We conclude by stating this result as a formal theorem, whose proof we have worked through during this section. 
\begin{theorem}
    Let $A$ be a real matrix. Define $\Phi(m,k)$ as in Eq.~\eqref{eq:wp_recursive}. Let $A_q=Ae^{i2\pi q}+A^{\top}e^{-i2\pi q}$ be the magnetic matrix with potential $q\in\mathbb{R}$. Then we can convert from $\Phi(m,k)$ to $A_q^m$ for any $q$ by Eq.\eqref{eq:wp2mw}, or reconstruct $\Phi(m,k)$ from $(A_{q_0}^m,...,A_{q_{\lfloor m/2\rfloor +1}}^m)$ with $q_{j}=\frac{j}{2(m+1)}$ via Eq.\eqref{eq:mw2wp}.
\end{theorem}


\subsection{Walk profile compression}

\subsubsection{Walk profile reconstruction algorithm}
We describe the complete pipeline of reconstructing walk profile from powers of magnetic matrix. Given a directed network $\mathcal{G}=(\mathcal{V},\mathcal{E})$ of node set $\mathcal{V}$ and edge set $\mathcal{E}\subseteq \mathcal{V}\times \mathcal{V}$, we first sample a subset of node set $\mathcal{V}_0$. We compute random-walk walk profiles $\widehat{\Phi}_{u,v}(m,k)$ and magnetic matrix $[\widehat{A}^m_{q_0}]_{u,v},...,[\widehat{A}_{q_{Q-1}}^m]_{u,v}$ for all $u\in\mathcal{V}_0$ and $v\in\mathcal{V}$. That is, we only evaluate the walks from source nodes $u$ in the sampled node set, to reduce the quadratic computational complexity. The calculation of walk profiles follow the recursive formula Eq.\eqref{eq:wp_recursive}:
\begin{equation}
    \forall u\in\mathcal{V}_0,v\in\mathcal{V}, \quad \widehat{\Phi}_{u,v}(m,k)=\sum_{w\in\mathcal{V}}\widehat{\Phi}_{u,w}(m-1,k-1)\widehat{A}_{w,v}+\widehat{\Phi}_{u,w}(m-1,k)\widehat{A}_{v,w}
\end{equation}
with $\widehat{A}=D^{-1}A$.
Similarly,
\begin{equation}
    \forall u\in\mathcal{V}_0,v\in\mathcal{V}, \quad [\widehat{A}^m_{q_j}]_{u,v}=\sum_{w\in \mathcal{V}}[\widehat{A}^{m-1}_{q_j}]_{u,w}[\widehat{A}_{q_j}]_{w,v}.
\end{equation}
After collecting walk profiles $\widehat{\Phi}(m,\cdot)$ and magnetic matrix $\widehat{A}_{q_0}^m,...,\widehat{A}_{q_{Q-1}}^m$, we solve the expanded linear system Eq.\eqref{eq:expand_linear_system} in real domain, by applying the pseudo inverse of the real Fourier matrix $\begin{pmatrix}
    \mathrm{Re}(F_{\vec{q}, m+1})\\
    \mathrm{Im}(F_{\vec{q}, m+1})
    \end{pmatrix}$ to magnetic matrix. This produces a solution serving as the reconstructed walk profiles denoted by $\widehat{\Phi}^{\mathrm{recon}}(m\cdot)$. Finally, we evaluate the reconstruction error by

\begin{equation}
    \mathrm{error}=\frac{1}{N\cdot N_0}\sum_{u\in\mathcal{V}_0,v\in\mathcal{V}}\frac{\norm{\widehat{\Phi}_{u,v}(m,\cdot)-\widehat{\Phi}^{\mathrm{recon}}_{u,v}(m,\cdot)}}{\max\{\norm{\widehat{\Phi}_{u,v}(m,\cdot)}, \norm{\widehat{\Phi}^{\mathrm{recon}}_{u,v}(m,\cdot)}\}},
\end{equation}
where $N=|\mathcal{V}|$ and $N_0=|\mathcal{V}_0|$.

\subsubsection{Decayed phenomenon of spectral radius}
\label{supp:decayed_spectral_radius}
Figure~\ref{fig:full_decayed_spectral_radius} demonstrates the decayed phenomenon of spectral radius of $\widehat{A}_q$ for all the network datasets we adopted. Note that ogbg-code2 datasets have no decayed phenomenon: the spectral radius of $\widehat{A}_q$ is always $1$ for any $q$. This is because these graphs are just a tree and there are no cycles. This leads to a uniform spectral density (as shown in Figures~\ref{fig:spectral_density_q} and~\ref {fig:spectral_density_q_len20}) and thus poor compressibility of walk profiles (shown in Figures~\ref{fig:recon_error_q} and~\ref{fig:recon_error_q_len20}). 
\begin{figure}[t!]
    \centering
    \includegraphics[width=1\linewidth]{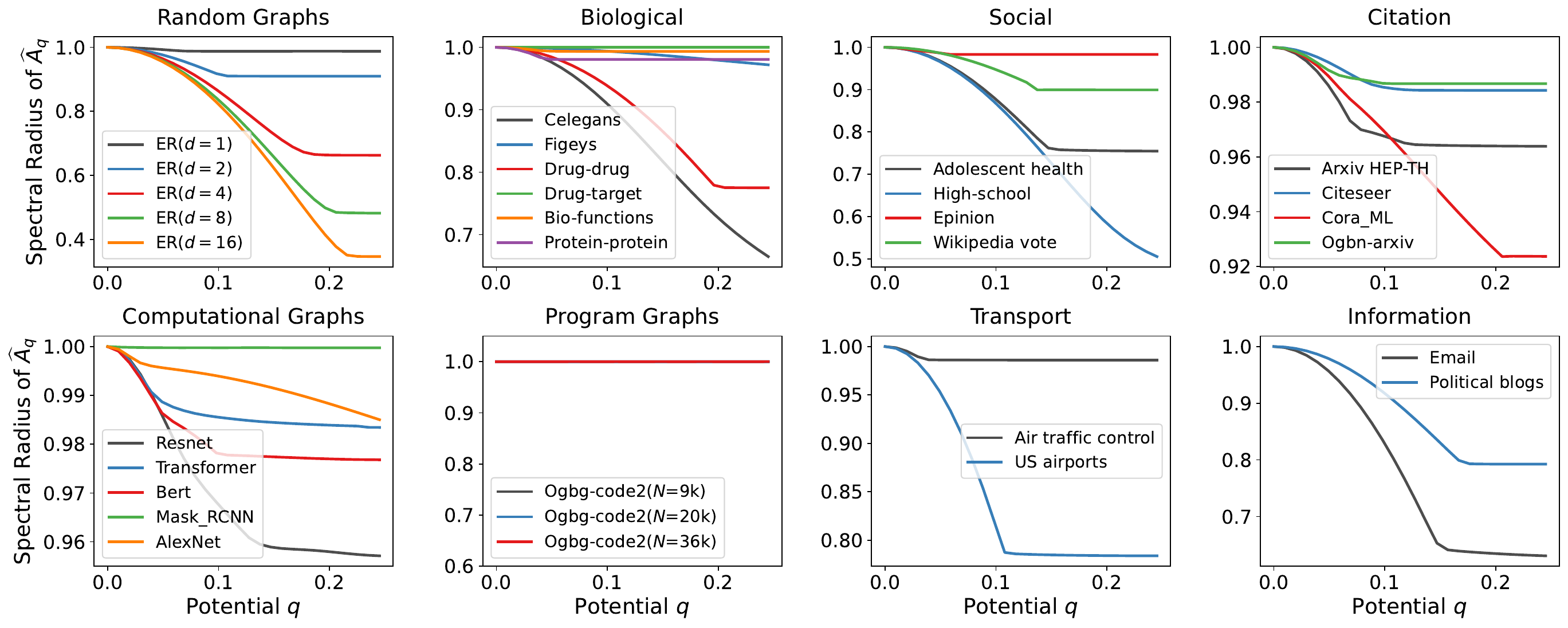}
    \caption{The decay phenomenon of spectral radius of $\widehat{A}_q$, w.r.t. the potential $q$ on networks from different domains.}
    \label{fig:full_decayed_spectral_radius}
\end{figure}

\subsubsection{Compression results of different walk length}
\label{supp:compression}
We investigated the compression of walk profiles of walk length $50$ in our main text, and here we show the results for walk length $20$ and that the conclusion is robust to the choice of walk length. Figure~\ref{fig:spectral_density_q_len20} illustrates the spectral density of $\widehat{\Phi}(20, \cdot)$ and Figure~\ref{fig:recon_error_q_len20} demonstrates the reconstruction errors. Again, we observe that if the spectral radius decays faster, then the spectral density is more concentrated on a few low-frequency components, and consequently, the reconstruction errors are overall smaller. We validate this correlation between spectral radius and walk profile compressibility of length $20$ again in Figure~\ref{fig:recon_error_spectral_radius_len20}.
\begin{figure}[t!]
    \centering
    \includegraphics[width=1\linewidth]{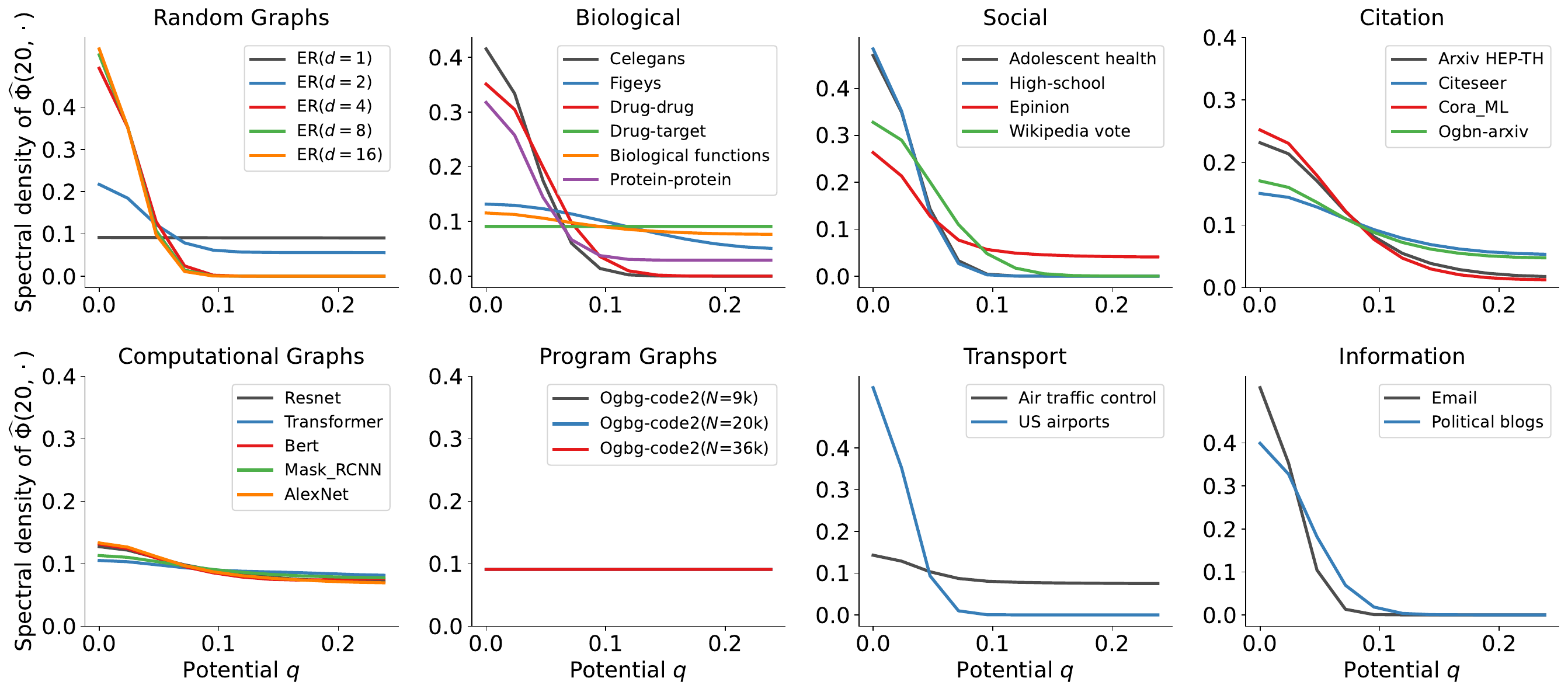}
    \caption{\textbf{Spectral density of walk profiles.} The relative spectral density of walk profiles $\Phi(20, \cdot)$, i.e., $\langle |A_q^{20}|^2\rangle$ with respect to potential $q$. The average $\langle \cdot\rangle$ is taken over all node pairs in the network.}
    \label{fig:spectral_density_q_len20}
\end{figure}
\begin{figure}[t!]
    \centering
    \includegraphics[width=1\linewidth]{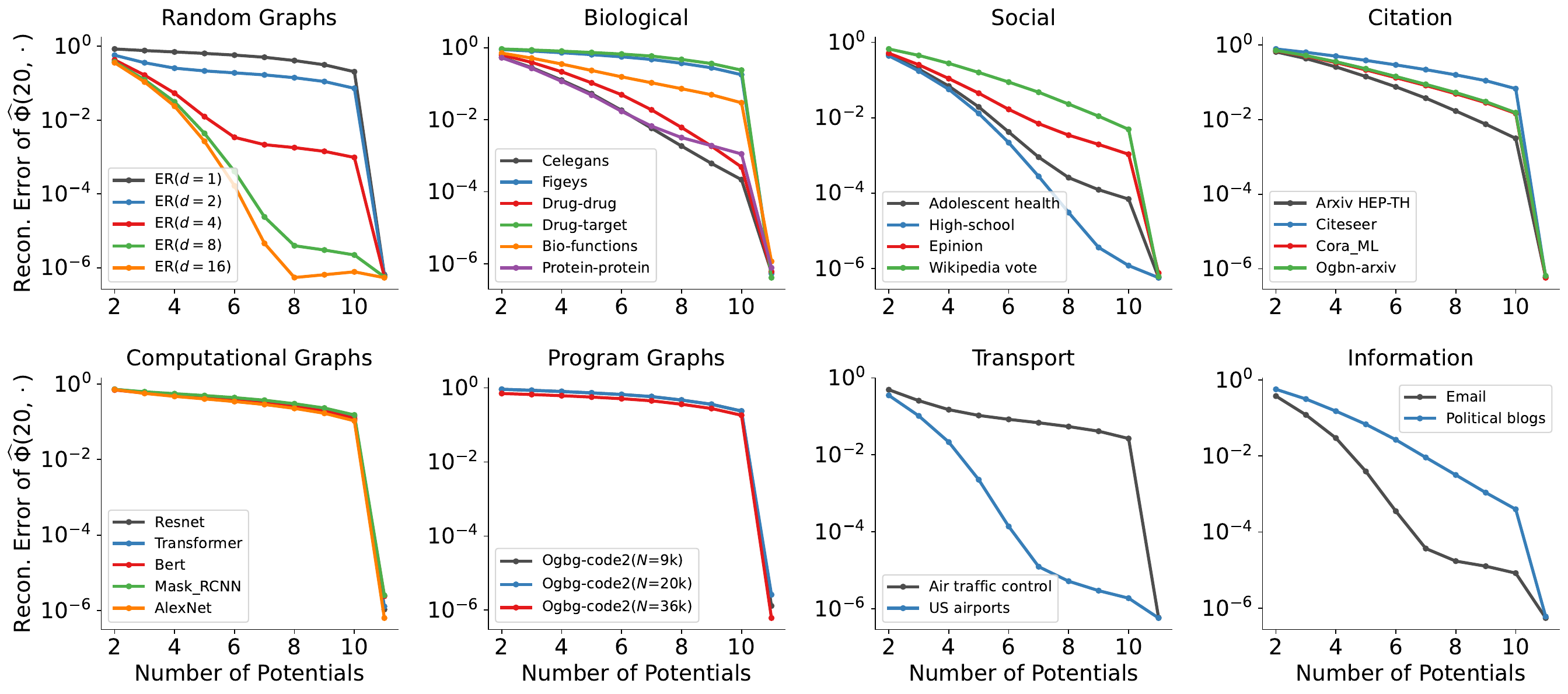}
    \caption{\textbf{Compressibility of walk profiles.} The average relative reconstruction error of walk profiles $\widehat{\Phi}(20, \cdot)$ computed using $[\widehat{A}_{q_0}^{20}, ..., \widehat{A}^{20}_{q_{Q-1}}]$ with varying numbers of the top-$Q$ smallest potentials, on a wide range of directed networks from different domains.}
    \label{fig:recon_error_q_len20}
\end{figure}
\begin{figure}
    \centering
    \includegraphics[width=0.6\linewidth]{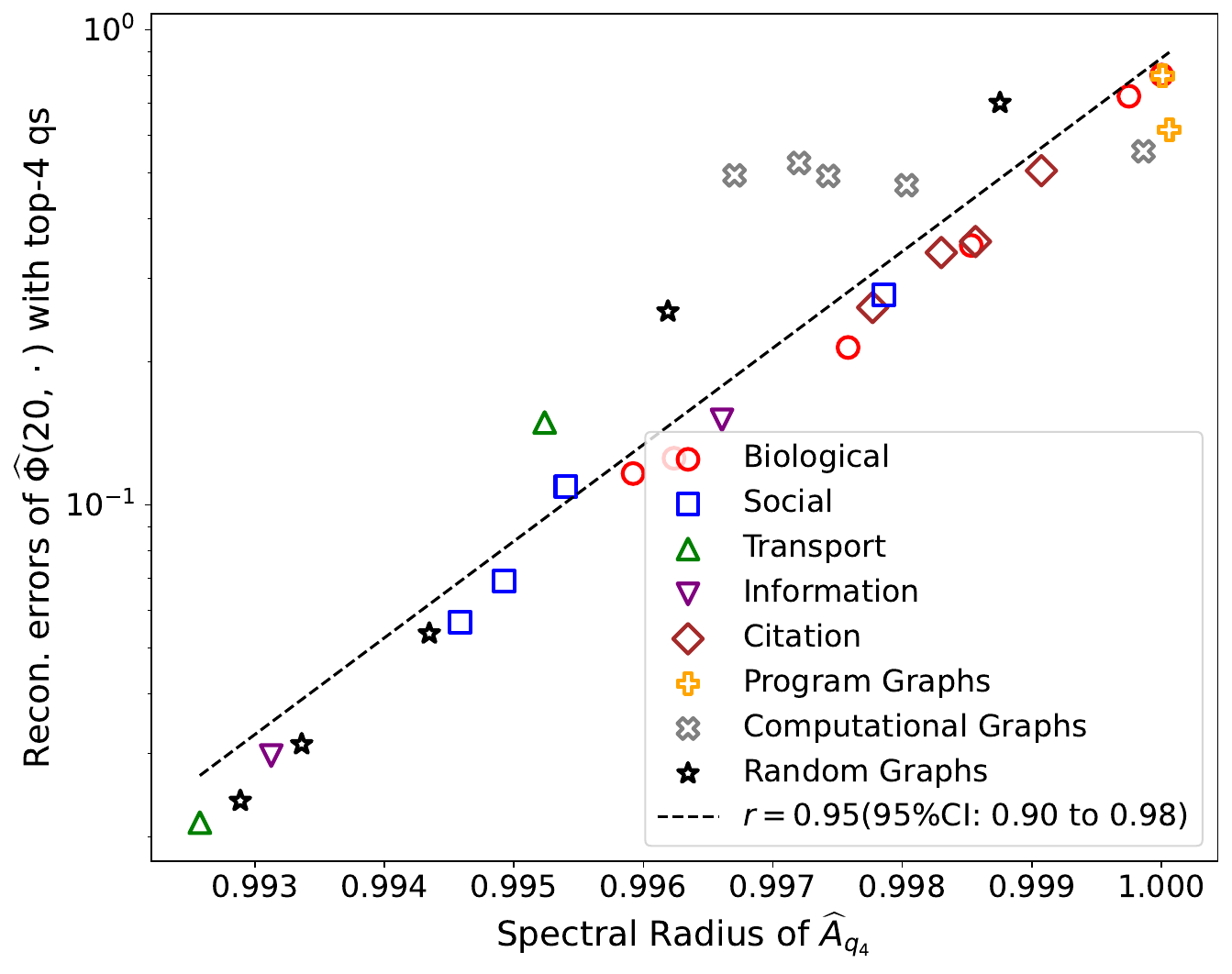}
    \caption{The decay rate of spectral radius (measured by spectral radius of $\widehat{A}_{q_{4}}$ with $q_{4}=\frac{4}{2(20+1)}=2/21$) is strongly correlated with the compressibility of walk profiles $\widehat{\Phi}(20, \cdot)$. The Pearson correlation coefficient is $r=0.95$ with 95\% confidence interval $(0.90, 0.98)$, estimated by Fisher transformation.}
    \label{fig:recon_error_spectral_radius_len20}
\end{figure}

\subsubsection{Hard examples for compression} 
\label{supp:hard_examples}
\begin{figure}[t!]
    \centering
    \includegraphics[width=0.96\linewidth]{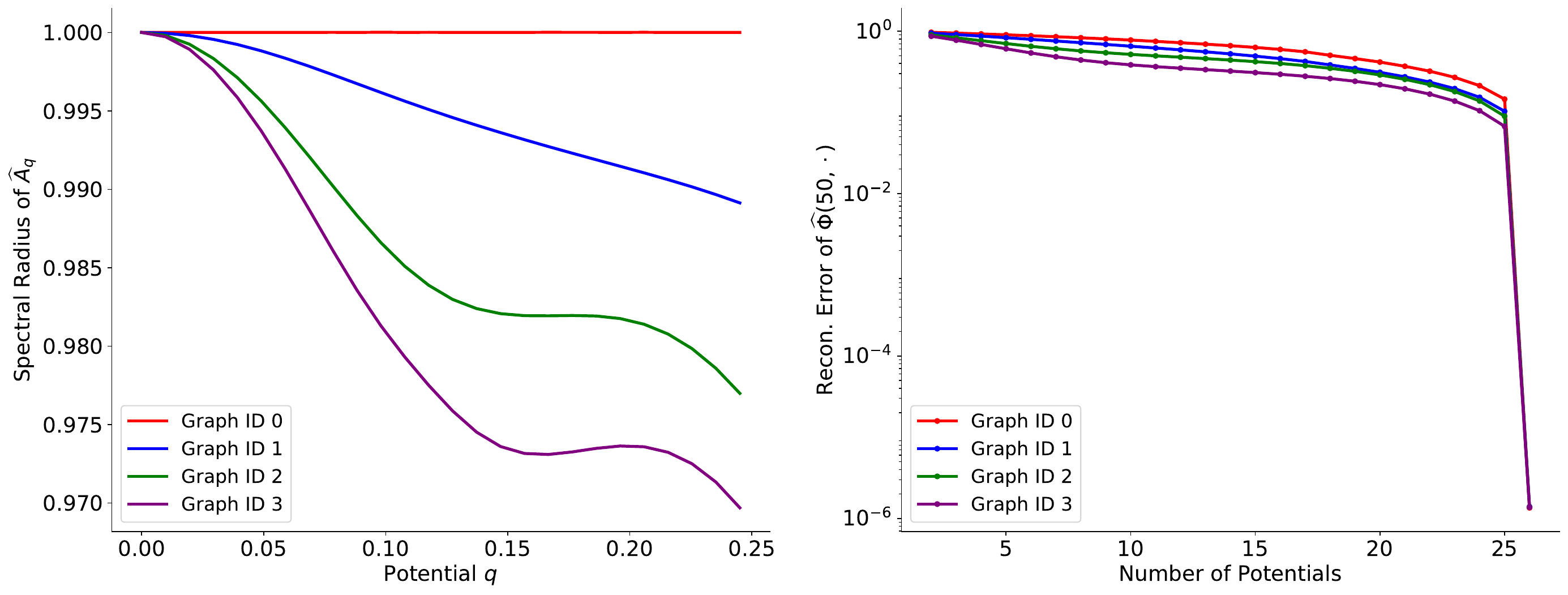}
    \caption{\textbf{A sequence of networks with decreasing spectral radius and increasing compressibility.} Left: the spectral radius of $\widehat{A}_q$ is gradually decreasing when we sequentially add more edges. Right: the compressibility becomes better for the same sequence of networks.}
    \label{fig:growing_graph}
\end{figure}
The connection between compressibility of walk profiles and spectral radius of the magnetic matrix allows us to understand how the graph structure affects the compression. Particularly, as the smallest eigenvalue of magnetic normalized Laplacian $I-D^{-1/2}A_qD^{-1/2}$ for a tree (directed graphs without cycles) is always $0$ regardless of $q$, the spectral radius of $\hat{A}_q=D^{-1}A_q$ is always $1$. Therefore, we expect trees to be the most difficult examples for walk profile compression. 

To further validate this, we first start with a directed tree graph from the code-2 dataset and sequentially add edges to make it gradually not tree-like. Let $\widehat{A}^{(t)}_{q_j}$ denote the random-walk magnetic matrix at step $t$ of the edge addition process, with potential $q_j=\frac{j}{2(m+1)}$ (we choose $m=50$). At each step $t$, we adapt~\cite{ghosh2006growing} and choose an edge to be added based on a heuristic that is expected to decrease the largest eigenvalue of $A_{q_{10}}^{(t)}$. We iteratively add edges following this heuristic until the average change of overall spectral radius $\frac{1}{\lfloor m/2\rfloor +1}\sum_{j=0}^{\lfloor m/2\rfloor +1}\rho(\widehat{A}^{(t)}_{q_j})-\rho(\widehat{A}^{(t+1)}_{q_j})$ is greater than a threshold $\tau=0.05$. At this point, we save a checkpoint for the perturbed graph, and we continue to add edges again. We stop adding edges until we obtain 3 checkpoints, which we label as graphs with IDs 1 through 3. The initial tree graph is referred to as ID 0. Figure \ref{fig:growing_graph} shows the reconstruction error on the these checkpointed graphs. We can see that compression is hardest for the tree graph in the beginning, but the compression becomes increasingly easier as the spectral radius gets decreases when we gradually add new edges. 

\subsection{Motif detection analysis}
\label{supp:motif}
In this section, we provide technical discussions about the capability of walk profiles for motif detection. Formally, given a directed graph $\mathcal{G}=(\mathcal{V},\mathcal{E})$ with node set $\mathcal{V}$ and edge set $\mathcal{E}\subseteq \mathcal{V}\times \mathcal{V}$, a target motif is a subgraph $\mathcal{G}_0=(\mathcal{V}_0, \mathcal{E}_0)$ where $\mathcal{V}_0\subseteq \mathcal{V}$ and $\mathcal{E}_0\subseteq \mathcal{E}$ contains nodes from $\mathcal{V}_0$ exclusively. The detection of target motif $\mathcal{G}_0$ is a function $f_{\mathcal{G}_0}:\mathcal{V}\to \{0,1\}$ defined by 
\begin{equation}
    f_{\mathcal{G}_0}(u)=1 ,\quad \text{if } u\in\mathcal{V}_0,
\end{equation}
and $f_{\mathcal{G}_0}(u)=0$ otherwise.

We focus on detecting cycles by walk profiles. Specifically, we extract $\Phi_{u,u}(m,\cdot)$ to express $f_{\mathcal{G}_0}(u)$, and $\mathcal{G}_0$ is a cycle with forward edges and backward edges. We adopt the same notation in the main text and represent a cycle by $C_{b_1b_2\cdots b_m}$, e.g., $C_{1011}$ represents a cycle $u_1\to u_2\leftarrow u_3\to u_4\to u_1$. Generally, it is hard to detect motifs by walks, since the walk sequence could involve repeated nodes. However, even assuming the walk sequence does not involve any repeated nodes, we can show that walk profiles cannot detect cycles with equal or more than two backward edges. 
\begin{theorem}
     Suppose $C_{b_1b_2\cdots b_m}$ is a cycle such that either at least two of $b_i$ is $0$ or at least two of $b_i$ is $1$. Let us define walk profile $\Phi_{u,u}(m,k)$ be the count of $m$-step bidirectional walks of $k$ forward edges, but without counting any walk sequences that involve repeated nodes. Then there exists two non-isomorphic graphs $\mathcal{G},\mathcal{G}^{\prime}$, and nodes $u\in\mathcal{V},u^{\prime}\in\mathcal{V}^{\prime}$, such that $f_{C_{b_1b_2\cdots b_m}}(u)=1$, $f^{\prime}_{C_{b_1b_2\cdots b_m}}(u^{\prime})=0$ but $\Phi_{u,u}(m,\cdot)=\Phi^{\prime}_{u^{\prime},u^{\prime}}(m,\cdot)$.
\end{theorem}
\begin{proof}
     Let $\mathcal{G}=C_{b_1b_2\cdots b_m}$ and $\mathcal{G}^{\prime}=C_{\sigma(b_1b_2\cdots b_m)}$, where $\sigma(b_1b_2\cdots b_m)$ is a permutation of $b_1b_2\cdots b_m$. We force this permutation to not be a cyclic permutation. This implies that $f_{C_{b_1b_2\cdots b_m}}(u)=1$ but $f^{\prime}_{C_{b_1b_2\cdots b_m}}(u^{\prime})=0$ for any nodes $u,u^{\prime}$. Now on $\mathcal{G}$, we have $\Phi_{u,u}(m,\sum_{i=1}^mb_i)=\Phi_{u,u}(m,m-\sum_{i=1}^mb_i)=1$ and $\Phi_{u,u}(m,k)=0$ for other $k$. This is the same for $\Phi^{\prime}_{u,u}(m,\cdot)$ on $\mathcal{G}^{\prime}$.
\end{proof}
On the other hand, the following result states that walk profiles without counting repeated nodes can detect cycles without backward edge (i.e., a fully-oriented directed cycle) or with one backward edge (e.g., feed-forward loops).
\begin{theorem}
     Suppose $C_{b_1b_2\cdots b_m}$ is a cycle such that either at most one $b_i$ is $0$ or at most one $b_i$ is $1$. Let us define walk profile $\Phi_{u,u}(m,k)$ be the count of $m$-step bidirectional walks of $k$ forward edges, but without counting any walk sequences that involve repeated nodes. Then $\Phi_{u,u}(m,\sum_{i=1}^mb_i)=f_{C_{b_1b_2\cdots b_m}}(u)$.
\end{theorem}
\begin{proof}
   Suppose $\sum_{i=1}b_i=0$ or $\sum_{i=1}^mb_i=m$, i.e., $C_{b_1b_2\cdots b_m}$ is a fully-oriented cycle. Then  by definition $\Phi_{u,u}(m,0)=\Phi_{u,u}(m,m)=f_{C_{b_1b_2\cdots b_m}}(u)$. Suppose $\sum_{i=1}b_i=1$ or $\sum_{i=1}^mb_i=m-1$, we can see that $f_{C_{b_1b_2\cdots b_m}}(u)=1\implies \Phi_{u,u}(m,1)=\Phi_{u,u}(m,m-1)>0$, and conversely $\Phi_{u,u}(m,1)=\Phi_{u,u}(m,m-1)>0\implies f_{C_{\sigma(b_1b_2\cdots b_m)}}(u)=1$, where $\sigma(b_1b_2\cdots b_m)$ is a permutation of $b_1b_2\cdots b_m$. Importantly, if either one $b_i$ is zero or one, $\sigma(b_1b_2\cdots b_m)$ is simply a cyclic permutation of $b_1b_2\cdots b_m$. This means
    $f_{C_{b_1b_2\cdots b_m}}(u)=f_{C_{b_1b_2\cdots b_m}}(u)$, due to the cyclic permutation symmetry of cycles. Thus, $\Phi_{u,u}(m,\sum_{i=1}^mb_i)=f_{C_{b_1b_2\cdots b_m}}(u)$ in this two cases.
\end{proof}